\documentclass[11pt,letterpaper]{article}
\usepackage{fullpage}
\usepackage{amsmath,amssymb,amsthm,mathtools,enumerate,multirow,url}
\usepackage{tikz}
\usepackage{bigstrut}
\usepackage{pgfplots}
\newtheorem{theorem}{Theorem}
\newtheorem{proposition}{Proposition}
\newtheorem{lemma}{Lemma}

\newtheorem{corollary}{Corollary}

\theoremstyle{definition}
\newtheorem{definition}{Definition}

\pgfplotscreateplotcyclelist{blackwhite}{%
loosely dotted, every mark/.append style={solid, fill=gray}, mark=triangle*\\%
dashed, every mark/.append style={solid, fill=white}, mark=square*\\%
dotted, every mark/.append style={solid, fill=gray},mark=star\\%
solid, every mark/.append style={solid, fill=gray}, mark=*\\%
densely dotted, every mark/.append style={solid, fill=gray}, mark=otimes*\\%
loosely dashed, every mark/.append style={solid, fill=gray},mark=*\\%
dasdotdotted, every mark/.append style={solid},mark=diamond*\\%
densely dashed, every mark/.append style={solid, fill=gray},mark=square*\\%
dashdotted, every mark/.append style={solid, fill=gray},mark=otimes*\\%
densely dashdotted,every mark/.append style={solid, fill=gray},mark=diamond*\\%
}

\newcommand{\Real}{\mathbb R}
\newcommand{\ceil}[1]{\left\lceil #1 \right\rceil}
\newcommand{\MM}[1]{\mathsf{MM}(#1)}
\newcommand{\Rank}[1]{\mathsf{Rank}(#1)}
\newcommand{\INV}[1]{\mathsf{INV}(#1)}
\newcommand{\DET}[1]{\mathsf{DET}(#1)}
\newcommand{\SYS}[1]{\mathsf{SYS}(#1)}
\newcommand{\DIS}[1]{\mathsf{DIST}(#1)}
\newcommand{\poly}{\mathrm{poly}}
\newcommand{\dist}{\ast}
\newcommand{\rk}{\mathrm{rank}}
\newcommand{\charpoly}{\mathbf{charpol}}
\newcommand{\minpoly}{\mathbf{minpol}}
\newcommand{\mybar}[1]{\lambda}

\providecommand{\field}{\mathbb{F}}

\providecommand{\Int}{\mathbb{Z}}

\title{Further Algebraic Algorithms in the Congested Clique Model and Applications to Graph-Theoretic Problems}
\author{
Fran{\c c}ois Le Gall\\
Graduate School of Informatics\\
Kyoto University\\
\url{legall@i.kyoto-u.ac.jp}}

\begin{document}
\date{}
\maketitle
\thispagestyle{empty}
\setcounter{page}{1}

\begin{abstract}
Censor-Hillel et al.~[PODC'15] recently showed how to efficiently implement centralized algebraic algorithms for matrix multiplication in the congested clique model, a model of distributed computing that has received increasing attention in the past few years. This paper develops further algebraic techniques for designing algorithms in this model.  We present deterministic and randomized algorithms,  in the congested clique model, for efficiently computing multiple independent instances of matrix products, computing the determinant, the rank and the inverse of a matrix, and solving systems of linear equations. As applications of these techniques, we obtain more efficient algorithms for the computation, again in the congested clique model, of the all-pairs shortest paths and the diameter in directed and undirected graphs with small weights, improving over Censor-Hillel et al.'s work. We also obtain algorithms for several other graph-theoretic problems such as computing the number of edges in a maximum matching and the Gallai-Edmonds decomposition of a simple graph, and computing a minimum vertex cover of a bipartite graph.
\end{abstract}

\newpage

\section{Introduction}
\paragraph{Background.}
The congested clique model is a model in distributed computing that has recently received increasing attention \cite{Censor-Hillel+15,Dolev+DISC12,Drucker+PODC14,Hegeman+PODC15,Hegeman+SIROCCO14,Hegeman+DISC14,Henzinger+15,LenzenPODC13,Lenzen+STOC11,Lotker+SPAA03,NanongkaiSTOC14,PattShamir+PODC11}. In this model $n$ nodes communicate with each other over a fully-connected network (i.e., a clique) by exchanging messages of size $O(\log n)$ in synchronous rounds. Compared with the more traditional congested model \cite{Peleg00}, the congested clique model removes the effect of distances in the computation and thus focuses solely on understanding the role of congestion in distributed computing.

Typical computational tasks studied in the congested clique model are graph-theoretic problems \cite{Censor-Hillel+15,Dolev+DISC12,Drucker+PODC14,Hegeman+PODC15,Henzinger+15,NanongkaiSTOC14}, where a graph $G$ on $n$ vertices is initially distributed among the $n$ nodes of the network (the $\ell$-th node of the network knows the set of vertices adjacent to the $\ell$-th vertex of the graph, and the weights of the corresponding edges if the graph is weighted) and the nodes want to compute properties of $G$. Besides their theoretical interest and potential applications, such problems have the following natural interpretation in the congested clique model: the graph~$G$ represents the actual topology of the network, each node knows only its neighbors but can communicate to all the nodes of the network, and the nodes want to learn information about the topology of the network. 

Censor-Hillel et al.~\cite{Censor-Hillel+15} recently developed algorithms for several graph-theoretic problems in the congested clique model by showing how to implement centralized algebraic algorithms for matrix multiplication in this model. More precisely, they constructed a $O(n^{1-2/\omega})$-round algorithm for matrix multiplication, where $\omega$ denotes the exponent of matrix multiplication (the best known upper bound on $\omega$ is $\omega<2.3729$, obtained in \cite{LeGallISSAC14,WilliamsSTOC12}, which gives exponent $1-2/\omega<0.1572$ in the congested clique model), improving over the $O(n^{2-\omega})$ algorithm mentioned in~\cite{Drucker+PODC14}, in the following setting: given two $n\times n$ matrices $A$ and $B$ over a field,  the $\ell$-th node of the network initially owns the $\ell$-th row of $A$ and the $\ell$-column of $B$, and needs to output the $\ell$-th row and the $\ell$-column of the product $AB$. 
Censor-Hillel et al.~consequently obtained $O(n^{1-2/\omega})$-round algorithms for several graph-theoretic tasks that reduce to computing the powers of (some variant of) the adjacency matrix of the graph, such as counting the number of triangles in a graph (which lead to an improvement over the prior best algorithms for this task \cite{Dolev+DISC12,Drucker+PODC14}), detecting the existence of a constant-length cycle and approximating the all-pairs shortest paths in the input graph (improving the round complexity obtained in \cite{NanongkaiSTOC14}). 
One of the main advantages of such an algebraic approach in the congested clique model is its versatility: it makes possible to construct fast algorithms for graph-theoretic problems, and especially for problems for which the best non-algebraic centralized algorithm is highly sequential and does not seem to be implementable efficiently in the congested clique model, simply by showing a reduction to matrix multiplication (and naturally also showing that this reduction can be implemented efficiently in the congested clique model).\vspace{-3mm} 

\paragraph{Our results.}
In this paper we develop additional algebraic tools for the congested clique model.

We first consider the task of computing in the congested clique model not only one matrix product, but multiple independent matrix products. More precisely, given $k$ matrices $A_1,\ldots,A_k$ each of size $n\times m$ and $k$ matrices $B_1,\ldots,B_k$ each of size $m\times m$, initially evenly distributed among the $n$ nodes of the network, the nodes want to compute the $k$ matrix products $A_1B_1,\ldots,A_kB_k$. Prior works \cite{Censor-Hillel+15,Drucker+PODC14} considered only the case $k=1$ and $m=n$, i.e., one product of two square matrices. Our contribution is thus twofold: we consider the rectangular case, and the case of several matrix products as well. Let us first discuss our results for square matrices ($m=n$). By using sequentially $k$ times the matrix multiplication algorithm from~\cite{Censor-Hillel+15}, $k$ matrix products can naturally be computed in $O(k n^{1-2/\omega})$ rounds. In this work we show that we can actually do better.
\begin{theorem}[Simplified version]\label{theorem:main}
In the congested clique model $k$ independent products of pairs of $n\times n$ matrices can be computed with round complexity
\[
\left\{
\begin{array}{ll}
O(k^{2/\omega}n^{1-2/\omega})&\textrm{ if }\:\: 1\le k< n,\\
O(k) &\textrm{ if }\:\:k\ge n.
\end{array}
\right.
\]
\end{theorem}
This generalization of the results from~\cite{Censor-Hillel+15} follows from a simple strategy: divide the $n$ nodes of the network into $k$ blocks (when $k\le n$), each containing roughly $n/k$ nodes, compute one of the $k$ matrix products per block by using an approach similar to \cite{Censor-Hillel+15} (i.e., a distributed version of the best centralized algorithm computing one instance of square matrix multiplication), and finally distribute the relevant part of the $k$ output matrices to all the nodes of the network.  Analyzing the resulting protocol shows that the dependence in $k$ in the overall round complexity is reduced to $k^{2/\omega}$. This sublinear dependence in $k$ has a significant number of implications (see below). 

The complete version of Theorem \ref{theorem:main}, given in Section \ref{sec:mm}, also considers the general case where the matrices may not be square (i.e., the case $m\neq n$), which will be crucial for some of our applications to the All-Pairs Shortest Path problem. The proof becomes more technical than for the square case, but is conceptually very similar: the main modification is simply to now implement a distributed version of the best centralized algorithm for rectangular matrix multiplication. The upper bounds obtained on the round complexity depend on the complexity of the best centralized algorithms for rectangular matrix multiplication (in particular the upper bounds given in \cite{LeGallFOCS12}). Figure \ref{fig2} depicts the upper bounds we obtain for the case $k=1$. While the major open problem is still whether the product of two square matrices can be computed in a constant (or nearly constant) number of rounds, our results show that for $m=O(n^{0.651\ldots})$, the product of an $n\times m$ matrix by an $m\times n$ matrix can indeed be computed in $O(n^\epsilon)$ rounds for any $\epsilon>0$. We also show lower bounds on the round complexity of the general case (Proposition \ref{prop:LB} in Section \ref{sec:mm}), which are tight for most values of $k$ and $m$, based on simple arguments from communication complexity.
 
\begin{figure}[ht]
\centering
%\tikzstyle{every pin}=[%fill=white,
%draw=black]%, font=\footnotesize]
\begin{tikzpicture}
\begin{axis}[
legend cell align=left,
width=11.5cm, 
height=5.5cm,
xmin=-0.1, xmax=2.1,
thick,
 scale only axis,
xmajorgrids,
ymajorgrids,
y label style={at={(-0.05,0.5)}},
 xtick={0,0.5,0.65149,1,2,3},
 xticklabels={$n^0$,$n^{0.5}$,$n^{0.651\ldots}\!\!\!\!\!\!\!\!\!\!\!$,$n^1$,$n^2$,$n^3$},
 ytick={0,0.15806,0.5,1,2},
 yticklabels={$n^0$,$n^{0.1572}$,$n^{0.5}$,$n^1$,$n^2$},
xlabel={Value of $m$},
ylabel={Round complexity},
%ylabel style={rotate=-90},
%cycle list name=blackwhite,
legend pos=south east]

\addplot [solid, every mark/.append style={solid, fill=gray}] coordinates {
(0,0)
(0.5,0)
(0.65149,0)
(0.655010867,3.1499E-05)
(0.660063058,0.000185466)
(0.665157209,0.00046928)
(0.670291956,0.000884717)
(0.675465707,0.001432944)
(0.7018152,0.006050667)
(0.728676702,0.013369825)
(0.755702037,0.022808146)
(0.771985579,0.029312812)
(0.782615312,0.03384583)
(0.80921228,0.046061402)
(0.835346497,0.059122839)
(0.860915298,0.072768653)
(0.885849284,0.086794269)
(0.910103875,0.10103875)
(0.933652685,0.115369134)
(0.956484808,0.12969616)
(0.978598694,0.143947753)
(1,0.158063833)
(1.040714109,0.185717816)
(1.07875901,0.2124099)
(1.114298646,0.238009024)
(1.147508562,0.26245719)
(1.178564031,0.285743876)
(1.247859649,0.339040935)
(1.30706033,0.385879339)
(1.358068003,0.427091196)
(1.402381239,0.463491682)
(1.475356113,0.524643887)
(1.532691149,0.573847081)
(1.57907065,0.6139529)
(1.648641356,0.675679322)
(2,1)
(3,2)
(3.1,2.1)
};

\addplot [only marks, every mark/.append style={solid, fill=gray}, mark=otimes*] coordinates {
(0,0)
(0.5,0)
(0.65149,0)
%(1,0.0)
(1,0.158063)
%(1,0.33333)
(2,1)
(3,2)
};

%\legend{with best bound on  $\omega(\gamma)$, with trivial bound on $\omega(\gamma)$}
%, Ideal algorithm}
\end{axis}
\end{tikzpicture}
\vspace{-4mm}
\caption{\label{fig2}Our upper bounds on the round complexity of the computation of the product of an $n\times m$ matrix by an $m\times n$ matrix in the congested clique model.}
\end{figure}
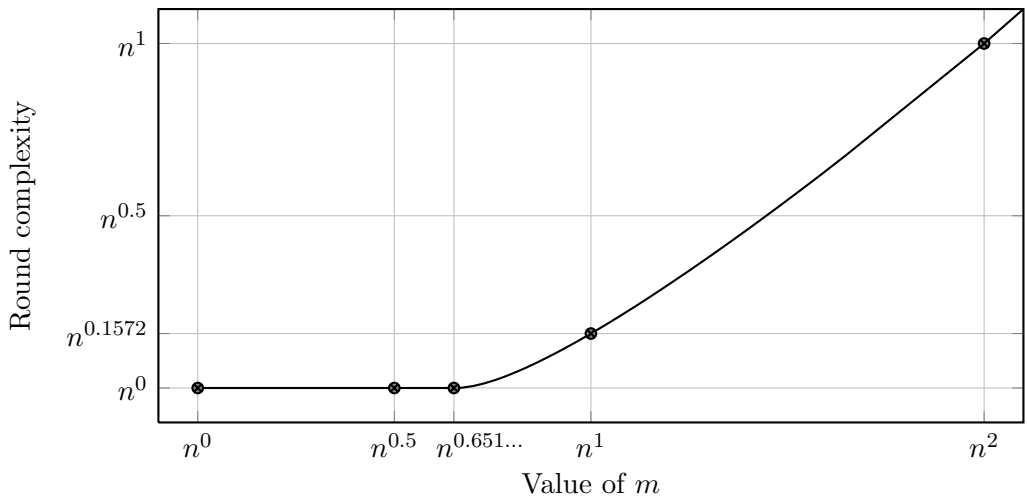

We then study the following basic problems in linear algebra: computing the determinant, the rank or the inverse of an $n\times n$ matrix over a finite field $\field$ of order upper bounded by a polynomial of $n$, and solving a system of $n$ linear equations and $n$ variables. We call these problems $\DET{n,\field}$, $\Rank{n,\field}$, $\INV{n,\field}$ and $\SYS{n,\field}$, respectively (the formal definitions are given in Section \ref{sec:prelim}).
While it is known that in the centralized setting these problems can be solved with essentially the same time complexity as matrix multiplication \cite{Burgisser+97}, these reductions are typically sequential and do not work in a parallel setting. In this paper we design fast deterministic and randomized algorithm for these four basis tasks, and obtain the following results. 
\begin{theorem}\label{theorem:det-inv}
Assume that $\field$ has characteristic greater than $n$. In the congested clique model, the deterministic round complexity of $\DET{n,\field}$ and $\INV{n,\field}$ is $O(n^{1-1/\omega})$.
\end{theorem}\vspace{-5mm}

\begin{theorem}\label{th:rand}
Assume that $\field$ has order $|\field|=\Omega(n^2\log n)$. In the congested clique model, the randomized round complexity of $\DET{n,\field}$, $\SYS{n,\field}$ and $\Rank{n,\field}$ is $O(n^{1-2/\omega}\log n)$.
\end{theorem}
The upper bounds of Theorems \ref{theorem:det-inv} and \ref{th:rand} are $O(n^{0.5786})$ and $O(n^{0.1572})$, respectively, by basing our implementation on the asymptotically fastest (but impractical) centralized algorithm for matrix multiplication corresponding to the upper bound $\omega<2.3729$. These bounds are $O(n^{2/3})$ and $ O(n^{1/3}\log n)$, respectively, by basing our implementation on the trivial (but practical) centralized algorithm for matrix multiplication (corresponding to the bound $\omega\le 3$). 
These algorithms are obtained by carefully adapting to the congested clique model the relevant known parallel algorithms \cite{Csanky+FOCS75,Kaltofen+SPAA91,Kaltofen+FOCS92,Kaltofen+91,Prerarata+IPL78} for linear algebra, and using our efficient algorithm for computing multiple matrix products (Theorem \ref{theorem:main}) as a subroutine. An interesting open question is whether $\INV{n,\field}$ can be solved with the same (randomized) round complexity as the other tasks. This problem may very well be more difficult; in the parallel setting in particular, to the best of our knowledge, whether matrix inversion can be done with the same complexity as these other tasks is also an open problem. \vspace{-6mm}

\paragraph{Applications of our results.}
The above results give new algorithms for many graph-theoretic problems in the congested clique model, as described below and summarized in Table \ref{table}. 

Our main key tool to derive these applications is Theorem \ref{th:dist} in Section \ref{sec:mm}, which gives an algorithm computing efficiently the distance product (defined in Section \ref{sec:prelim}) of two matrices with small integer entries based on our algorithm for multiple matrix multiplication of Theorem~\ref{theorem:main}. Computing the distance product is a fundamental graph-theoretic task deeply related to the All-Pairs Shortest Path (APSP) problem \cite{SeidelJCSS95,Shoshan+FOCS99,ZwickJACM02}. Combining this result with techniques from \cite{Shoshan+FOCS99}, and observing that these techniques can be implemented efficiently in the congested clique model, we then almost immediately obtain the following result. 
\begin{theorem}\label{th:APSPu}
In the congested clique model, the deterministic round complexity of the all-pairs shortest paths problem in an undirected graph of $n$ vertices with integer weights in $\{0,\ldots,M\}$, where~$M$ is an integer such that $M\le n$, is $\tilde O(M^{2/\omega}n^{1-2/\omega})$.
\end{theorem}
Since computing the diameter of a graph reduces to solving the all-pairs shortest paths, we obtain the same round complexity for diameter computation in the same class of graphs. This  improves over the $\tilde O(Mn^{1-2/\omega})$-round algorithm for these tasks (implicitly) given in \cite{Censor-Hillel+15}. The main application of our results nevertheless concerns the all-pair shortest paths problem over directed graphs (for which the approach based on \cite{Shoshan+FOCS99} does not work) with constant weights. We obtain the following result by combining our algorithm for distance product computation with Zwick's approach \cite{ZwickJACM02}.
\begin{theorem}\label{th:APSPd}
In the congested clique model, the randomized round complexity of the all-pairs shortest paths problem in a directed graph of $n$ vertices with integer weights in $\{-M,\ldots,0,\ldots,M\}$, where~$M=O(1)$, is $O(n^{0.2096})$.
\end{theorem}
\noindent Prior to this work, the upper bound for the round complexity of this problem was $\tilde O(n^{1/3})$, obtained by directly computing the distance product (as done in \cite{Censor-Hillel+15}) in the congested clique model. Again, Theorem \ref{th:APSPd} follows easily from Theorem \ref{th:dist} and the observation that the reduction to distance product computation given in \cite{ZwickJACM02} can be implemented efficiently in the congested clique model. The exponent $0.2096$ in the statement of Theorem \ref{th:APSPd} is derived from the current best upper bounds on the complexity of rectangular matrix multiplication in the centralized setting \cite{LeGallFOCS12}.

Theorems \ref{theorem:det-inv} and \ref{th:rand} also enable us to solve a multitude of graph-theoretic problems in the congested clique model with a sublinear number of rounds. Examples described in this paper are computing the number of edges in a maximum matching of a simple graph with $O(n^{1-2/\omega}\log n)$ rounds, computing the set of allowed edges in a perfect matching, the Gallai-Edmonds decomposition of a simple graph, and a minimum vertex cover in a bipartite graph with $O(n^{1-1/\omega})$ rounds. These results are obtained almost immediately from the appropriate reductions to matrix inversion and similar problems known the centralized setting \cite{CheriyanSICOMP97,Lovasz79,Rabin+89} --- indeed it is not hard to adapt all these reductions so that they can be implemented efficiently in the congested clique model. 
Note that while non-algebraic centralized algorithms solving these problems also exist (see, e.g., \cite{Lovasz+09}), they are typically sequential and do not appear to be efficiently implementable in the congested clique model. The algebraic approach developed in this paper, made possible by our algorithms for the computation of the determinant, the rank and the inverse of matrix, appears to be currently  the only way of obtaining fast algorithms for these problems in the congested clique model.%\vspace{-3mm} 

\begin{table}[tb]
\label{table}
\begin{center}
\caption{Summary of the applications of our algebraic techniques to graph-theoretic problems in the congested clique model. Here $n$ both represents the number of vertices in the input graph and the number of nodes in the network.}
\vspace{2mm}
\begin{tabular}{|l|l|l|}\hline
Problem&Round complexity&Previously\\ \hline
APSP (undirected graphs, weights in $\{0,1,\ldots,M\}$)&$\tilde O\!\left(M^{\frac{2}{\omega}}n^{1-\frac{2}{\omega}}\right)$ \hspace{1mm}Th.~\ref{th:APSPu}\bigstrut&$\tilde O\left(Mn^{1-\frac{2}{\omega}}\right)$$\!\!\!$\\
APSP (directed graphs, constant weights)&$O(n^{0.2096})$ \hspace{7mm}$\,$Th.~\ref{th:APSPd}\bigstrut&$\tilde O(n^{1/3})$\\
Diameter (undirected graphs, weights in $\{0,1,\ldots,M\}$)$\!\!\!$&$\tilde O\!\left(M^{\frac{2}{\omega}}n^{1-\frac{2}{\omega}}\right)$ \hspace{1mm}Cor.~\ref{cor:diameter}\bigstrut&$\tilde O\left(Mn^{1-\frac{2}{\omega}}\right)$$\!\!\!$\\
Computing the size of a maximum matching &$O\big(n^{1-\frac{2}{\omega}}\log n\big)$ \hspace{1mm}Th.~\ref{th:mm}\bigstrut&$\:\:\:\:$---\\
Computing allowed edges in a perfect matching&$O(n^{1-1/\omega})$ \hspace{7mm}Sec.~\ref{sub:match}&$\:\:\:\:$---\\
Gallai-Edmonds decomposition&$O(n^{1-1/\omega})$ \hspace{7mm}Th.~\ref{th:GE}&$\:\:\:\:$---\\
Minimum vertex cover in bipartite graphs&$O(n^{1-1/\omega})$ \hspace{7mm}Sec.~\ref{sub:GE}&$\:\:\:\:$---\\\hline
\end{tabular}
\end{center}
%\vspace{-5mm}
\end{table}

\section{Preliminaries}\label{sec:prelim}
\paragraph{Notations.}Through this paper we will use $n$ to denote the number of nodes in the network. The~$n$ nodes will be denoted $1,2,\dots,n$. The symbol $\field$ will always denote a finite field of order upper bounded by a polynomial in $n$ (which means that each field element can be encoded with $O(\log n)$ bits and thus sent using one message in the congested clique model). Given any positive integer $p$, we use the notation $[p]$ to represent the set $\{1,2,\ldots,p\}$. Given any $p\times p'$ matrix $A$, we will write its entries as $A[i,j]$ for $(i,j)\in[p]\times [p']$, and use the notation $A[i,\ast]$ to represent its $i$-th row and $A[\ast,j]$ to represent its $j$-th column.\vspace{-3mm}
%\subsection{The congested clique model}
\paragraph{Graph-theoretic problems in the congested clique model.}%\label{sub:graph}
As mentioned in the introduction, typically the main tasks that we want to solve in the congested clique model are graph-theoretical problems. In all the applications given in this paper the number of vertices of the graph will be $n$, the same as the number of nodes of the network. The input will be given as follows: initially each node $\ell\in[n]$ has the $\ell$-th row and the $\ell$-th column of the adjacency matrix of the graph. Note that this distribution of the input, while being the most natural, is not essential; the only important assumption is that the entries are evenly distributed among the~$n$ nodes since they can then be redistributed in a constant number of rounds as shown in the following Lemma by Dolev et al.~\cite{Dolev+DISC12}, which we will use many times in this paper. 
\begin{lemma}{\cite{Dolev+DISC12}}\label{lemma:Dolev}
In the congested clique model a set of messages in which no node is the source of more than $n$ messages and no node is the destination of more than $n$ messages can be delivered within two rounds if the source and destination of each message is known in advance to all nodes.  
\end{lemma}\vspace{-5mm}

\paragraph{Algebraic problems in the congested clique model.}
The five main algebraic problems that we consider in this paper are defined as follows. \vspace{3mm}

%%%%%%%%%%%%%%%
\noindent$\MM{n,m,k,\field}$ --- Multiple Rectangular Matrix Multiplications  \\\vspace{-4mm}

\noindent\hspace{3mm} Input: matrices $A_{1},\ldots, A_{k}\in\field^{n\times m}$ and $B_{1},\ldots, B_{k}\in \field^{m\times n}$ distributed among the $n$ nodes\\
\vspace{-5mm}

\noindent\hspace{16mm}(Node $\ell\in[n]$ has $A_{1}[\ell,\ast],\ldots, A_{k}[\ell,\ast]$ and $B_{1}[\ast,\ell],\ldots, B_{k}[\ast,\ell]$)\\
\vspace{-4mm}

\noindent\hspace{3mm} Output: the matrices $A_{1}B_{1},\ldots,A_{k}B_{k}$ distributed among the $n$ nodes\\\vspace{-5mm}

\noindent\hspace{19mm}(Node $\ell\in[n]$ has $A_{1}B_{1}[\ell,\ast],\ldots, A_{k}B_{k}[\ell,\ast]$ and $A_{1}B_{1}[\ast,\ell],\ldots, A_{k}B_{k}[\ast,\ell]$)\\\vspace{-2mm}

%%%%%%%%%%%%%%%%

\noindent$\DET{n,\field}$ --- Determinant  \\\vspace{-4mm}

\noindent\hspace{3mm} Input: matrix $A\in\field^{n\times n}$ distributed among the $n$ nodes (Node $\ell\in[n]$ has $A[\ell,\ast]$ and $A[\ast,\ell]$)\\
\vspace{-5mm}

\noindent\hspace{3mm} Output: $\det(A)$ \hspace{3mm} (Each node of the network has $\det(A)$)\\\vspace{-2mm}

%%%%%%%%%%%%%%%%
\noindent$\Rank{n,\field}$ --- Rank  \\\vspace{-4mm}

\noindent\hspace{3mm} Input: matrix $A\in\field^{n\times n}$ distributed among the $n$ nodes (Node $\ell\in[n]$ has $A[\ell,\ast]$ and $A[\ast,\ell]$)\\
\vspace{-5mm}

\noindent\hspace{3mm} Output: $\rk(A)$\hspace{3mm} (Each node of the network has $\rk(A)$)\\\vspace{-2mm}

%%%%%%%%%%%%%%%
\noindent$\INV{n,\field}$ --- Inversion  \\\vspace{-4mm}

\noindent\hspace{3mm} Input: invertible matrix $A\in\field^{n\times n}$ distributed among the $n$ nodes \\\vspace{-5mm}

\noindent\hspace{16mm}
 (Node $\ell\in[n]$ has $A[\ell,\ast]$ and $A[\ast,\ell]$)\\
\vspace{-5mm}

\noindent\hspace{3mm} Output: matrix $A^{-1}$ distributed among the $n$ nodes  (Node $\ell\in[n]$ has $A^{-1}[\ell,\ast]$ and $A^{-1}[\ast,\ell]$)\\\vspace{-5mm}

\noindent\hspace{18mm}
 (Node $\ell\in[n]$ has $A^{-1}[\ell,\ast]$ and $A^{-1}[\ast,\ell]$)\\\vspace{-2mm}

%%%%%%%%%%%%%
\noindent$\SYS{n,\field}$ --- Solution of a linear system  \\\vspace{-4mm}

\noindent\hspace{3mm} Input: invertible matrix $A\in\field^{n\times n}$ and vector $b\in\field^{n\times 1}$, distributed among the $n$ nodes\\\vspace{-5mm}

\noindent\hspace{16mm}
 (Node $\ell\in[n]$ has $A[\ell,\ast]$, $A[\ast,\ell]$ and $b$)\\
\vspace{-5mm}

\noindent\hspace{3mm} Output: the vector $x\in\field^{n\times 1}$ such that $Ax=b$\hspace{3mm} (Node $\ell\in[n]$ has $x[\ell]$)\\

Note that the distribution of the inputs and the outputs assumed in the above five problems is mostly chosen for convenience. For instance, if needed the whole vector $x$ in the output of $\SYS{n,\field}$ can be sent to all the nodes of the network in two rounds using  Lemma \ref{lemma:Dolev}. The only important assumption is that when dealing with matrices, the entries of the matrices must be evenly distributed among the~$n$ nodes.

We will also in this paper consider the distance product of two matrices, defined as follows. 

\begin{definition}
Let $m$ and $n$ be two positive integers. Let $A$ be an $n\times m$ matrix and $B$ be an $m\times n$ matrix, both with entries in $\Real\cup\{\infty\}$. The distance product of $A$ and $B$, denoted $A\dist B$, is the $n\times n$ matrix $C$ such that $C[i,j]=\min_{s\in[m]}\{A[i,s]+B[s,j]\}$ for all $(i,j)\in [n]\times [n]$. 
\end{definition}

We will be mainly interested in the case when the matrices have integer entries. More precisely, we will consider the following problem. \vspace{3mm}

%%%%%%%%%%%%%
\noindent$\DIS{n,m,M}$ --- Computation of the distance product  \\\vspace{-4mm}

\noindent\hspace{3mm} Input: an $n\times m$ matrix $A$ and an $m\times n$ matrix $B$, with entries in $\{-M,\ldots,-1,0,1,\ldots,M\}\cup\{\infty\}$\\\vspace{-5mm}

\noindent\hspace{15mm}
(Node $\ell\in[n]$ has $A[\ell,\ast]$ and $B[\ast,\ell]$)\\
\vspace{-5mm}

\noindent\hspace{3mm} Output: the matrix $C=A\ast B$ distributed among the $n$ nodes\\\vspace{-5mm}

\noindent\hspace{18mm}
 (Node $\ell\in[n]$ has $C[\ell,\ast]$ and $C[\ast,\ell]$)\\
\vspace{-6mm}

\paragraph{Centralized algebraic algorithms for matrix multiplication.}
We now briefly describe algebraic algorithms for matrix multiplication and known results about the complexity of rectangular matrix multiplication. We refer to \cite{Burgisser+97} for a detailed exposition of these concepts.

Let $\field$ be a field and $m,n$ be two positive integer. Consider the problem of computing the product of an $n\times m$ matrix by an $m\times n$ matrix over $\field$.
An algebraic algorithm for this problem is described by three sets $\{\alpha_{ij\mu}\}$, $\{\beta_{ij\mu}\}$ and $\{\lambda_{ij\mu}\}$ of coefficients from $\field$ such that, for any $n\times m$ matrix $A$ and any $m\times n$ matrix $B$, the equality
\[
C[i,j]=\sum_{\mu=1}^t\lambda_{ij\mu}S^{(\mu)}T^{(\mu)}
\]
holds
for all $(i,j)\in[n]\times[n]$,
where $C=AB$ and 
\[
S^{(\mu)}=\sum_{i=1}^n\sum_{j=1}^m\alpha_{ij\mu} A[i,j],
\hspace{10mm}
T^{(\mu)}=\sum_{i=1}^n\sum_{j=1}^m\beta_{ij\mu} B[j,i],
\]
for each $s\in[t]$. Note that each $S^{(\mu)}$ and each $T^{(\mu)}$ is an element of $\field$. The integer $t$ is called the rank of the algorithm, and corresponds to the complexity of the algorithm. 

For instance, consider the trivial algorithm computing this matrix product using the formula
\[
C[i,j]=\sum_{s=1}^m A[i,s]B[s,j].
\]
This algorithm can be described in the above formalism by taking $t=n^2m$, writing each $\mu\in[n^2m]$ as a triple $\mu=(i',j',s')\in[n]\times [n]\times [m]$, and choosing
\begin{align*}
\lambda_{ij(i',j',s')}&=\left\{
\begin{array}{cl}
1&\textrm{ if $i=i'$ and $j=j'$}, \\ 
0& \textrm{ otherwise},
\end{array}
\right.\\
\alpha_{ij(i',j',s')}&=\left\{
\begin{array}{cl}
1&\textrm{ if $i=i'$ and $j=s'$}, \\
0& \textrm{ otherwise},
\end{array}
\right.
\hspace{15mm}
\beta_{ij(i',j',s')}=\left\{
\begin{array}{cl}
1&\textrm{ if $i=j'$ and $j=s'$}, \\ 
0& \textrm{ otherwise}.
\end{array}
\right.
\end{align*}
Note that this trivial algorithm, and the description we just gave, also works over any semiring. 

\paragraph{The exponent of matrix multiplication.}
For any non-negative real number $\gamma$, let $\omega(\gamma)$ denote the minimal value $\tau$
such that the product of an $n\times \ceil{n^\gamma}$ matrix over $\field$ by an $\ceil{n^\gamma}\times n$ matrix over $\field$ can be computed by an algebraic algorithm of rank $n^{\tau+o(1)}$ (i.e., can be computed with complexity $O(n^{\tau+\epsilon})$ for any $\epsilon>0$). As usual in the literature, we typically abuse notation and simply write that such a product can be done with complexity $O(n^{\omega(\gamma)})$, i.e., ignoring the $o(1)$ in the exponent. The value $\omega(1)$ is denoted by $\omega$, and often called the exponent of square matrix multiplication.
Another important quantity is 
the value $\alpha=\sup\{\gamma\:|\:\omega(\gamma)=2\}$. 

The trivial algorithm for matrix multiplication gives the upper bound $\omega(\gamma)\le 2+\gamma$, and thus  $\omega\le 3$ and $\alpha\ge 0$.  The current best upper bound on $\omega$ is $\omega<2.3729$, see \cite{LeGallISSAC14,WilliamsSTOC12}.
The current best  bound on $\alpha$ is $\alpha>0.3029$, see~\cite{LeGallFOCS12}. 
The best bounds on $\omega(\gamma)$ for $\gamma>\alpha$ can also be found in \cite{LeGallFOCS12}.%\vspace{2mm}

%%%%%%%%%%%%%%%%%%%%%%%%%%%%%%%%%%%%%%%%%%%%%%%%%%%%%%%%%
\section{Matrix Multiplication in the Congested Clique Model}\label{sec:mm}
In this section we present our results on the round complexity of $\MM{n,m,k,\field}$ and $\DIS{n,m,M}$. 

We first give the complete statement of our main result concerning $\MM{n,m,k,\field}$ that was stated in a simplified form in the introduction. 
\addtocounter{theorem}{-5}
\begin{theorem}[Complete version]
For any positive integer $k\le n$, the deterministic round complexity of $\MM{n,m,k,\field}$ is 
\[
\left\{
\begin{array}{ll}
O(k)&\textrm{ if }\:\: 0\le m\le \sqrt{kn},\\
O(k^{2/\omega(\gamma)}n^{1-2/\omega(\gamma)})&\textrm{ if }\:\:\sqrt{kn}\le m< n^2/k,\\
O(km/n) &\textrm{ if }\:\:m\ge n^2/k,
\end{array}
\right.
\]
where $\gamma$ is the solution of the equation
\begin{equation}\label{eq:maincond}
\left(1-\frac{\log k}{\log n}\right)\gamma=1-\frac{\log k}{\log n}+\left(\frac{\log m}{\log n}-1\right)\omega(\gamma).
\end{equation}
For any $k\ge n$, the deterministic round complexity of $\MM{n,m,k,\field}$ is
\[
\left\{
\begin{array}{ll}
O(k)&\textrm{ if }\:\: 1\le m\le n,\\
O(km/n) &\textrm{ if }\:\:m\ge n.
\end{array}
\right.
\]
\end{theorem}
\addtocounter{theorem}{+5}
The proof of Theorem \ref{theorem:main}, which will also show that Equation~(\ref{eq:maincond}) always has a solution when $k\le n$ and $\sqrt{kn}\le m< n^2/k$, is given in Section \ref{sub:proof} (a short discussion of the proof ideas was presented in the introduction). The upper bounds we obtain for the case $k=1$ are depicted in Figure~\ref{fig2}, where Equation (\ref{eq:maincond}) is solved using the best known upper bound on $\omega(\gamma)$ from \cite{LeGallFOCS12}. As briefly mentioned in the introduction, the round complexity is constant for any $k\le \sqrt{n}$, and we further have round complexity $O(n^{\epsilon})$, for any $\epsilon>0$, for all values $k\le n^{(1+\alpha)/2}$ (the bound $\alpha>0.3029$ implies $(1+\alpha)/2> 0.6514$). For the case $m=n$ the solution of Equation (\ref{eq:maincond}) is $\gamma=1$, which gives the bounds of the simplified version of Theorem \ref{theorem:main} presented in the introduction.

We now give lower bounds on the round complexity of $\MM{n,m,k,\field}$ that show that the upper bounds of Theorem \ref{theorem:main} are tight, except possibly in the case $\sqrt{kn}\le m < n^2/k$ when $k\le n$.
\begin{proposition}\label{prop:LB}
The randomized round complexity of $\MM{n,m,k,\field}$ is 
\[
\left\{
\begin{array}{ll}
\Omega(k)&\textrm{ if }\:\: 1\le m\le n,\\
\Omega(km/n) &\textrm{ if }\:\:m\ge n.
\end{array}
\right.
\]
\end{proposition}
\begin{proof}
We first prove the lower bound $\Omega(km/n)$ for any $m\ge n$.
Let us consider instances of $\MM{n,m,k,\field}$ of the following form: for each $s\in[k]$ all the rows of $A_{s}$ are zero except the first row; for each $s\in[k]$ all the columns of $B_{s}$ are zero except the second column. Let us write $C_s=A_sB_s$ for each $s\in[k]$. We prove the lower bound by partitioning the $n$ nodes of the network into the two sets $\{1\}$ and $\{2,\ldots,n\}$, and considering the following two-party communication problem. Alice (corresponding to the set $\{1\}$) has for input $A_{s}[1,j]$ for all $j\in[m]$ and all $s\in[k]$. Bob (corresponding to the set $\{2,\ldots,n\}$) has for input $B_{s}[i,2]$ for all $i\in[m]$ and all $s\in[k]$. The goal is for Alice to output $C_{s}[1,2]$ for all $s\in[k]$. Note that $C_{s}[1,2]$ is the inner product (over $\field$) of the first row of $A_{s}$ and the second column of $B_s$. Thus $\sum_{s=1}^k C_{s}[1,2]$ is the inner product of two vectors of size $km$. Alice and Bob must exchange $\Omega(km\log|\field|)$ bits to compute this value \cite{Chu+95}, which requires $\Omega(km/n)$ rounds in the original congested clique model.

We now prove the lower bound $\Omega(k)$ for any $m\ge 1$. Let us consider instances of $\MM{n,m,k,\field}$ of the following form: for each $s\in[k]$, all entries of $A_{s}$ are zero except the entry $A_{s}[1,1]$ which is one; for each $s\in[k]$, $B_{s}[i,j]=0$ for all $(i,j)\notin \{(1,j)\:|\: j\in\{2,\ldots,n\}\}$ (the other $n-1$ entries are arbitrary). Again, let us write $C_s=A_sB_s$ for each $s\in[k]$. We prove the lower bound by again partitioning the $n$ nodes of the network into the two sets $\{1\}$ and $\{2,\ldots,n\}$, and considering the following two-party communication problem. Alice has no input. Bob has for input $B_{s}[1,j]$ for all $j\in\{2,\ldots,n\}$ and all $s\in[k]$. The goal is for Alice to output $C_{s}[1,j]$ for all $j\in\{2,\ldots,n\}$ and all $s\in[k]$. Since the output reveals Bob's whole input to Alice, Alice must receive $\Omega(k(n-1)\log|\field|)$ bits, which gives round complexity $\Omega(k)$ in the original congested clique model.
\end{proof}

\subsection{Proof of Theorem \ref{theorem:main}}\label{sub:proof}
Let us first prove the following proposition that deals with the case where $m$ is large.  In this case the algorithm is relatively simple. 
%The proof of Theorem \ref{theorem:main} will be based on the following two propositions.
\begin{proposition}\label{prop:large-l}
For any $k\le n$ and any $m\ge n^2/k$, 
the deterministic round complexity of $\MM{n,m,k,\field}$ is 
%there exists a distributed algorithm that computes $n^\alpha\cdot\mm{n,n^\ell,n}$ %with round complexity
$
O(km/n).
$
\end{proposition}
\begin{proof}
We assume below for convenience that both $n/k$ and $km/n$ are integers. If this is not the case the proof can be adjusted in a straightforward manner by replacing them by $\ceil{n/k}$ and $\ceil{km/n}$, respectively. 

For each $s\in [k]$, we will write $C_s=A_sB_s$.
Let us decompose the matrix $A_{s}$ into $n/k$ matrices of size $n\times \frac{km}{n}$ by partitioning the $m$ columns of $A_{s}$ into $n/k$ consecutive blocks of size $km/n$. Let us call these smaller matrices $A_{s}^{(1)}, \ldots, A_{s}^{(n/k)}$. Similarly, for each $s\in [k]$ we decompose the matrix $B_{s}$ into $n/k$ matrices of size $\frac{km}{n}\times n$ by partitioning the $m$ rows of $B_{s}$ into $n/k$ consecutive blocks of size $km/n$. Let us call these smaller matrices $B_{s}^{(1)}, \ldots, B_{s}^{(n/k)}$. For each $t\in[n/k]$ we write $C_s^{(t)}=A_s^{(t)}B_s^{(t)}$.

For each $s\in[k]$, and each $i,j\in[n]$ the equation
\begin{equation}\label{eq1}
C_s[i,j]=\sum_{t=1}^{n/k}C_s^{(t)}[i,j]
\end{equation}
obviously holds. Our distributed algorithm is based on this simple observation. Each node of the network will be assigned, besides its original label $\ell\in[n]$, a second label $(s,t)\in[k]\times[n/k]$. The assignment of these labels is arbitrary, the only condition being that distinct nodes are assigned distinct labels.  The distributed algorithm is as follows.
\begin{itemize}
\item[1.]
Node $(s,t)\in[k]\times [n/k]$ receives the whole two matrices $A_{s}^{(t)}$ and $B_{s}^{(t)}$ from the nodes of the networks owning the entries of these two matrices, and then locally computes $C_{s}^{(t)}$.
\item[2.]
Node $\ell\in[n]$ receives $C_{s}^{(t)}[\ell,\ast]$ and $C_{s}^{(t)}[\ast,\ell]$ for all $s\in[k]$ and all $t\in[n/k]$ from the nodes of the network owning these entries, and then locally computes $C_s[\ell,\ast]$ and $C_s[\ast,\ell]$ using Equation (\ref{eq1}) for all $s\in[k]$.
\end{itemize}
The number of field elements received per node is $2km$ at Step 1 and $2n^2$ at Step 2.
The total number of field elements received per node is thus $O(km+n^2)$, which gives round complexity $O(km/n)$ when $km\ge n^2$ from Lemma \ref{lemma:Dolev}. 
\end{proof}

The main technical contribution is the following proposition.
\begin{proposition}\label{prop:medium-l}
For any $k$ and $m$ such that $k\le n$ and $\sqrt{kn}\le m< n^2/k$,
the deterministic round complexity of $\MM{n,m,k,\field}$ is 
\[
O(k^{2/\omega(\gamma)}n^{1-2/\omega(\gamma)}),
\]
where $\gamma$ is the solution of the equation
$
\big(1-\frac{\log k}{\log n}\big)\gamma=1-\frac{\log k}{\log n}+\big(\frac{\log m}{\log n}-1\big)\omega(\gamma).
$
\end{proposition}
\begin{proof}
For convenience, let us assume that $n/k$ is an integer (otherwise we replace this value by the nearest integer). Let $\gamma\ge 0$ be a value that will be set later. Let $d$ be the largest integer such that the product of a $d\times \ceil{d^{\gamma}}$ matrix by a $\ceil{d^{\gamma}}\times d$ matrix can be computed by an algebraic algorithm (as in Section \ref{sec:prelim}) of rank $n/k$. Note that 
%$d\le\sqrt{n/k}$, and 
\[
d=\Theta\left((n/k)^{1/\omega(\gamma)}\right)
\]
from the definition of the exponent of matrix multiplication. For convenience, we assume below that $n/d$, $\sqrt{n/k}$ and $\sqrt{kn}/d$ are also integers (otherwise we can again simply replace these values by the nearest integer). Define the quantity
\begin{align*}
r&=\frac{m}{d^\gamma\sqrt{n/k}}=\Theta\left(\frac{m}{(n/k)^{\gamma/\omega(\gamma)+1/2}}\right).
\end{align*}
Our choice of $\gamma$ (discussed later) will guarantee that $m\ge (n/k)^{\gamma/\omega(\gamma)+1/2}$, which implies $r=\Omega(1)$. We will thus assume below, for convenience but without affecting the analysis of the complexity of the algorithm, that $r$ and $d^\gamma$ are integers such that $r\ge 1$ and $d^\gamma\le m$, and that $m/d^\gamma$ is an integer as well.

For each $s\in[k]$, we write the entries of $A_{s}$ and $B_{s}$ as $A_{s}[ix,jz]$ and $B_{s}[jz,ix]$, respectively, where $i\in[d]$, $j\in[d^\gamma]$, $x\in[n/d]$ and $z\in[m/d^\gamma]$. More precisely, this notation corresponds to decomposing each row of $A_s$ and each column of $B_s$ into $d$ consecutive blocks of $n/d$ entries, and decomposing each column of $A_s$ and each row of $B_s$ into $d^\gamma$ consecutive blocks of $m/d^\gamma$ entries. Note that this corresponds to decomposing $A_s$ into an $d\times d^\gamma$ matrix where each entry is a submatrix of~$A_s$ of size $(n/d)\times (m/d^\gamma)$, and decomposing $B_s$ into an $d^\gamma\times d$ matrix where each entry is a submatrix of~$B_s$ of size $(m/d^\gamma) \times (n/d)$. We will write $A_s[i\ast,j\ast]$ and $B_s[j\ast,i\ast]$  to represent these submatrices. Similarly, we write the elements of the output matrix $C_{s}=A_{s}B_{s}$ as $C_{s}[ix,jy]$, where $i,j\in[d]$ and $x,y\in[n/d]$. By extending the formulation given in Section \ref{sec:prelim} to block matrices (and using the same notations), for any $s\in[k]$, any $i,j\in[d]$ and any $x,y\in[n/d]$ we can write
\begin{equation}\label{eq:eq5}
C_s[ix,jy]=\sum_{\mu=1}^{n/k} \lambda_{ij\mu} S_s^{(\mu)}T_s^{(\mu)},
\end{equation}
where, for each $\mu\in[n]$, $S_s^{(\mu)}$ is an $(n/d)\times (m/d^\gamma)$ matrix that can be written as a linear combination of the submatrices of $A_s$, and $T_s^{(\mu)}$ is an $(m/d^\gamma)\times (n/d)$ matrix that can be written as a linear combination of the submatrices of $B_s$:
\begin{align}\label{eq:eq6}
S_s^{(\mu)}&=\sum_{i=1}^{d}\sum_{j=1}^{d^\gamma} \alpha_{ij\mu} A_s[i\ast,j\ast]\\\label{eq:eq7}
T_s^{(\mu)}&=\sum_{i=1}^{d}\sum_{j=1}^{d^\gamma} \beta_{ij\mu} B_s[j\ast,i\ast].
\end{align}
Finally, we will also decompose each label $x$, $y$ and~$z$ into two parts:
\begin{itemize} 
\item
we write $x=wx'$ where $w\in[\sqrt{n/k}]$ and $x'\in[\sqrt{kn}/d]$, 
\item
we write $y=wy'$ where $w\in[\sqrt{n/k}]$ and $y'\in[\sqrt{kn}/d]$, 
\item 
we write $z=wz'$ 
where $w\in[\sqrt{n/k}]$ and $z'\in[r]$.
\end{itemize}
This corresponds to further decomposing each submatrix of $A_s$ into an $\sqrt{n/k}\times \sqrt{n/k}$ matrix where each entry is a smaller submatrix of $A_s$ of size $(\sqrt{kn}/d)\times r$, and decomposing each submatrix of $B_s$ into an $\sqrt{n/k}\times \sqrt{n/k}$ matrix where each entry is a smaller submatrix of $B_s$ of size $r \times (\sqrt{kn}/d)$.

Each node of the network has a label $(i,x)\in[d]\times [n/d]$, and receives as input $A_s[ix,\ast]$ and $B_s[\ast,ix]$ for all $s\in[k]$. We assign two additional labels to each node: one label $(s,u,v)\in[k]\times[\sqrt{n/k}]\times [\sqrt{n/k}]$ and one label $(s,\mu)\in[k]\times [n/k]$. The assignment of these labels is arbitrary, the only condition being that distinct nodes are assigned distinct labels. The algorithm is given in Figure~\ref{fig:algorithmA}. 

\begin{figure}[t!]
\begin{center}
\fbox{
\begin{minipage}{15 cm} 
\begin{itemize}
\item[1.]
Node $(s,u,v)\in[k]\times[\sqrt{n/k}]\times [\sqrt{n/k}]$ receives $A_s[iux',jvz']$ and $B_s[juz',ivy']$ from the nodes of the network owning these entries, for all $i\in[d]$, $j\in[d^\gamma]$, $x',y'\in[\sqrt{kn}/d]$ and $z'\in[r]$.
%\item[]
Node $(s,u,v)$ then locally computes $S_s^{(\mu)}[ux',vz']$ using Equation (\ref{eq:eq6}) and $T_s^{(\mu)}[uz',vy']$ using Equation (\ref{eq:eq7}) for all $\mu\in[n/k]$, all $x',y'\in[\sqrt{kn}/d]$ and all $z'\in[r]$.
\item[2.]
Node $(s,\mu)\in[k]\times [n/k]$ receives the whole matrices $S_s^{(\mu)}$ and $T_s^{(\mu)}$ from the nodes of the network owning the entries of these matrices, and locally computes the product $P_s^{(\mu)}=S_s^{(\mu)}T_s^{(\mu)}$.
\item[3.]
Node $(s,u,v)\in[k]\times[\sqrt{n/k}]\times [\sqrt{n/k}]$ receives $P_s^{(\mu)}[ux',vy']$ from the nodes of the network owning these entries, for all $\mu\in[n/k]$ and all $x',y'\in[\sqrt{kn}/d]$.
%\item[]
Node $(s,u,v)$ then locally computes $C_s[iux',jvy']$ using Equation (\ref{eq:eq5}) for all $i,j\in[d]$ and all $x',y'\in[\sqrt{kn}/d]$. 
\item[4.]
Node $(i,x)\in[d]\times [n/d]$ receives $C_s[ix,\ast]$ and $C_s[\ast,ix]$ for all $s\in[k]$ from the nodes of the network owning these entries. 
\end{itemize}
\end{minipage}
}
\end{center}\vspace{-4mm}
\caption{Distributed algorithm for $\MM{n,m,k,\field}$ in the congested clique model. Initially each node $(i,x)\in[d]\times [n/d]$ has as input $A_s[ix,\ast]$ and $B_s[\ast,ix]$ for all $s\in[k]$.}\label{fig:algorithmA}
\end{figure}

The number of field elements received per node at Step 1 is
\[
2\times d\times d^\gamma\times \frac{\sqrt{kn}}{d}\times \frac{m}{d^\gamma\sqrt{n/k}}=2km.
\]
The number of field elements received per node at Step 2 is 
\begin{equation}\label{eq2}
2\times \frac{n}{d}\times \frac{m}{d^\gamma}=
O\left(k^{(1+\gamma)/\omega(\gamma)}mn^{1-(1+\gamma)/\omega(\gamma)}\right).
\end{equation}
The number of field elements received per node at Step 3 is 
\begin{equation}\label{eq3}
\frac{n}{k}\times \frac{\sqrt{kn}}{d}\times \frac{\sqrt{kn}}{d}=
O\left(k^{2/\omega(\gamma)}n^{2-2/\omega(\gamma)}\right).
\end{equation}
The number of field elements received per node at Step 4 is $2kn$. 
Note that Expression (\ref{eq2}) is larger than  $2km$ since $\omega(\gamma)\ge 1+\gamma$ and $k\le n$. Moreover, Expression (\ref{eq3}) is larger than $2kn$ since $\omega(\gamma)\ge 2$ and $k\le n$. Thus the total number of field elements received per node is 
\[
O\left(
k^{(1+\gamma)/\omega(\gamma)}mn^{1-(1+\gamma)/\omega(\gamma)}
+k^{2/\omega(\gamma)}n^{2-2/\omega(\gamma)}
\right).
\]
 
Let us write $a=\log k/\log n$ and $b=\log m/\log n$ and define the two functions 
\begin{align*}
f(\gamma)&=\frac{a(1+\gamma)}{\omega(\gamma)}+b+1-\frac{1+\gamma}{\omega(\gamma)},\\
g(\gamma)&=\frac{2a}{\omega(\gamma)}+2-\frac{2}{\omega(\gamma)},
\end{align*}
for any $\gamma\ge 0$. The function $f$ is a decreasing function of $\gamma$, with value $a/2+b+1/2$ when $\gamma=0$ and with limit $a+b$ when $\gamma$ goes to infinity. The function $g$ is an increasing function of $\gamma$, with value $a+1$ when $\gamma=0$ and with limit $2$ when $\gamma$ goes to infinity. Let us choose $\gamma$ as follows. Since we assumed $\sqrt{kn}\le m<n^2/k$, the equation $f(\gamma)=g(\gamma)$ necessarily has a solution. We choose $\gamma$ as this solution, i.e., such that 
\begin{equation}\label{eq:inter}
\left(1-a\right)\gamma=1-a+\left(b-1\right)\omega(\gamma).
\end{equation}
Observe that such a choice for $\gamma$ implies that $m\ge (n/k)^{\gamma/\omega(\gamma)+1/2}$ as we required: Equation (\ref{eq:inter}) can be rewritten as
\[
b=1+\frac{(1-a)\gamma}{\omega(\gamma)}-\frac{1-a}{\omega(\gamma)},
\]
which implies 
\[
b\ge (1-a)\left(\frac{\gamma}{\omega(\gamma)}+\frac{1}{2}\right)
\]
since $\omega(\gamma)\ge 2$.
For our choice of $\gamma$, the total number of field elements received per node is thus
$O(k^{2/\omega(\gamma)}n^{2-2/\omega(\gamma)})$, 
which gives round complexity
\[
O(k^{2/\omega(\gamma)}n^{1-2/\omega(\gamma)}),
\]
as claimed, from Lemma \ref{lemma:Dolev}.
\end{proof}

The proof of Theorem \ref{theorem:main} now follows easily from Propositions \ref{prop:large-l} and \ref{prop:medium-l}.
\begin{proof}[Proof of Theorem \ref{theorem:main}]
Consider first the case $k\le n$. 
When $\sqrt{kn}\le m\le n^2/k$, we use the algorithm of Proposition \ref{prop:medium-l}. When $m\ge n^2/k$, we use the algorithm of Proposition \ref{prop:large-l}.
When $0\le m\le \sqrt{kn}$ the claimed upper bound $O(k)$ on the round complexity can be derived from Proposition \ref{prop:large-l} by taking $m=\sqrt{kn}$ (i.e., by appending rows and columns with zero entries to the input matrices) and $\gamma=0$.

For the case $k\ge n$ we can simply repeat $\ceil{k/n}$ times an algorithm for $\MM{n,m,n,\field}$, which gives round complexity
\[
\left\{
\begin{array}{ll}
O(k)&\textrm{ if }\:\: 1\le m\le n,\\
O(km/n) &\textrm{ if }\:\:m\ge n.
\end{array}
\right.
\]
from the analysis of the previous paragraph.
\end{proof}

\subsection{Application to the distance product}
One of the main applications of Theorem \ref{theorem:main} is the following result, which will be the key ingredient for all our results on the all-pairs shortest paths and diameter computation discussed in Section \ref{sec:appli}. 
\begin{theorem}\label{th:dist}
For any $M\le n$ and $m\le n$, the deterministic round complexity of $\DIS{n,m,M}$ is 
\[
\left\{
\begin{array}{ll}
O(M\log m)&\textrm{ if }\:\: 0\le m\le \sqrt{Mn\log m},\\
O\left((M\log m)^{2/\omega(\gamma)}n^{1-2/\omega(\gamma)}\right)&\textrm{ if }\:\:\sqrt{Mn \log m}\le m\le n^2/(M\log m),\\
O\left(mM\log m/n\right)&\textrm{ if }\:\:n^2/(M\log m)\le m\le n,
\end{array}
\right.
\]
where $\gamma$ is the solution of the equation
$
\left(1-\frac{\log M}{\log n}\right)\gamma=1-\frac{\log M}{\log n}+\left(\frac{\log m}{\log n}-1\right)\omega(\gamma).
$
\end{theorem}
Let us first give a brief overview of the proof of Theorem \ref{th:dist}.
The idea is to show that $\DIS{n,m,M}$ reduces to $\MM{n,m,k,\field}$ for $k\approx M\log m$ and a well-chosen finite field $\field$, and then use Theorem~\ref{theorem:main} to get a factor $(M\log m)^{2/\omega(\gamma)}$, instead of the factor $M$ obtained in a straightforward implementation of the distance product, in the complexity. This reduction is done by first applying a standard encoding of the distance product into a usual matrix product of matrices with integer entries of absolute value $\exp(M)$, and then using Fourier transforms to split this latter matrix product into roughly~$M\log m$ independent matrix products over a small field. 
\begin{proof}[Proof of Theorem \ref{th:dist}]
We show below that $\DIS{n,m,M}$ reduces to $\MM{n,m,k,\field}$ where $k=O(M\log m)$ and $|\field|=\poly(M,\log m)$. The result then follows from Theorem \ref{theorem:main}.

Let $N$ be any positive integer. We first show how to reduce the multiplication or the addition of two nonnegative $N$-bit integers $a$ and $b$ to $2N$ independent operations (multiplications or additions, respectively) in a large enough field, using the Fourier transform. Let $\field$ be a finite field with characteristic at least $N+1$. Let 
\[
\Phi_N\colon \{0,1,\ldots,2^{N-1}\}\to \field[x]/(x^{2N}-1)
\]
be the map that maps each $N$-bit integer $c=\sum_{i=0}^{N-1}c_i 2^{i}$, with each $c_i$ in $\{0,1\}$, to the polynomial $\sum_{i=0}^{N-1}c_i x^{i}$. Computing the product $ab\in \Int$ (resp.~the sum $a+b\in \Int$) reduces to computing the product $\Phi_N(a)\Phi_N(b)$ (resp.~the sum $\Phi_N(a)+\Phi_N(b)$) over $\field[x]/(x^{2N}-1)$. Indeed, if $\Phi_N(a)\Phi_N(b)$ is the polynomial $\sum_{i=0}^{2N-1}c_i x^{i}$, with each $c_i$ in $\field$, then $ab$ is equal to the sum $\sum_{i=0}^{2N-1}c_i 2^{i}$ computed over the integers (note that the assumption on the characteristic of $\field$ is crucial here), and similarly for the addition. The computation of $\Phi_N(a)\Phi_N(b)$ and $\Phi_N(a)+\Phi_N(b)$ can be done using the identities
\begin{align}\label{eq:DFT1}
\Phi_N(a)\Phi_N(b)&=DFT^{-1}(DFT(\Phi_N(a))\cdot DFT(\Phi_N(b))),\\
\Phi_N(a)+\Phi_N(b)&=DFT^{-1}(DFT(\Phi_N(a))+ DFT(\Phi_N(b)))\label{eq:DFT2},
\end{align}
where $DFT\colon \field[x]/(x^{2N}-1) \to \field^{2N}$ is the Fourier transform and where $\cdot$ and the second $+$ represent the coordinate-wise multiplication and addition in $\field^{2N}$, respectively (we refer to \cite{Burgisser+97} for a detailed presentation of this Fourier transform). 
In order for the Fourier transform to be defined we nevertheless need to choose the field $\field$ such that it contains a $2N$-th primitive root of unity. It is known (see, e.g., \cite{De+SICOMP13}) that for any prime $p$ such that $2N$ divides $p-1$ the finite field $\field=\Int_p$ contains such a prime root. It is also known (\cite{Linnik44A,Linnik44B}, see also \cite{De+SICOMP13}) that the least prime $p$ such that $2N$ divides $p-1$ is smaller than $d(2N)^{d'}$ for some absolute constants $d$ and $d'$. By choosing such a prime $p$, we obtain a reduction from the computation of $ab$ (resp.~$a+b$) to one coordinate-wise multiplication (resp.~addition) in $\field^{2N}$ where $|\field|=\poly(N)$. Note that we do not need to discuss the costs of finding $p$ and the cost of preprocessing/postprocessing operations (such as applying the Fourier transform and its inverse), since they will have no impact on the round complexity of the algorithm we design below. 

Let us now consider the task of computing the integer $\sum_{t=1}^m a_tb_t$ given non-negative $N$-bit integers $a_1,\ldots,a_m,b_1,\ldots,b_m$. 
Similarly to what we just did, this sum can be recovered from the polynomial
$\sum_{t=1}^m\Phi_{N}(a_t)\Phi_{N}(b_t)$, which can be obtained
by computing the $2N$-dimentional vector
\begin{equation}\label{eq:DFT3}
\sum_{t=1}^m DFT(\Phi_{N}(a_t))\cdot DFT(\Phi_{N}(b_t))
\end{equation}
over a finite field $\field$ of order polynomial in $N$ and $m$, and then applying the inverse Fourier transform.

We can now describe our reduction. Let $A,B$ be the matrices with entries in $\{-M,\ldots,M\}\cup\{\infty\}$ of which we want to compute the distance product. We first reduce this distance product to one usual product of matrices with large entries, using standard techniques \cite{Alon+97,Takaoka98,ZwickJACM02}. Consider the $n\times m$ matrix $A'$ by the $m\times n$ matrix $B'$ defined as:
\begin{equation*}%\label{eq:dist}
A'[i,j]=\left\{\!\!
\begin{array}{cl}
(m+1)^{M-A[i,j]}&\textrm{if } A[i,j]\neq\infty,\\
0&\textrm{if } A[i,j]=\infty,
\end{array}
\right.
\hspace{5mm}
B'[j,i]=\left\{\!\!
\begin{array}{cl}
(m+1)^{M-B[j,i]}&\textrm{if } B[j,i]\neq\infty,\\
0&\textrm{if } B[j,i]=\infty,
\end{array}
\right.
\end{equation*}
for all $(i,j)\in [n]\times [m]$. It is easy to check (see  \cite{ZwickJACM02} for a proof) that the entry $A\dist B[i,j]$ can be recovered easily from entry $A'B'[i,j]$, for each $(i,j)\in [n]\times [m]$. 
Let $N=\ceil{\log_2((m+1)^{2M}+1)}$ be the number of bits needed to represent the entries of $A'$ and $B'$, and $\field$ be a finite field as in the previous paragraph. Next, in order to compute $A'B'$ we use the following strategy. For each $(i,j)\in[n]\times[m]$, consider the two vectors $DFT(\Phi_{N}(A'[i,j]))\in\field^{2N}$ and $DFT(\Phi_{N}(B'[j,i]))\in\field^{2N}$. For convenience write them as $\vec{p}_{ij}=(p_{ij}[1],\ldots,p_{ij}[2N])$ and $\vec{q}_{ji}=(q_{ji}[1],\ldots,q_{ji}[2N])$, respectively. Now, for any $s\in[2N]$, define the matrix $A'_s\in\field^{n\times m}$ and the matrix $B'_s\in \field^{m\times n}$ such that $A'_s[i,j]=p_{ij}[s]$ and $B_s[j,i]=q_{ji}[s]$ for all $(i,j)\in[n]\times [m]$. It follows from the discussion of the previous paragraph (and in particular Equation (\ref{eq:DFT3})) that for each $(i,j)\in[n]\times[n]$ the entry $A'B'[i,j]$ can be recovered from the entries $A'_1B'_1[i,j],\ldots, A'_{2N}B'_{2N}[i,j]$. Since all preprocessing and postprocessing steps of this strategy can be performed locally by the nodes of the network in the congested clique model, this reduces the computation of $A'B'$ to solving one instance of $\MM{n,m,2N,\field}$, with $2N=O(M\log m)$ and $|\field|=\poly(M,m)$, as claimed.
\end{proof}

%%%%%%%

%%%%%%%%%%%%%%%%%%%%%%%%%%%%%%%%%%%%%%%%%%%%%%%%%%%%%%%%%%%%%
\section{Deterministic Computation of Determinant and Inverse Matrix}
In this section we present deterministic algorithms for computing the determinant of a matrix and the inverse of a matrix in the congested clique model, and prove Theorem \ref{theorem:det-inv}. Our algorithms can be seen as efficient implementations of the parallel algorithm by Prerarata and Sarwate \cite{Prerarata+IPL78} based on the Faddeev-Leverrier method.

Let $A$ be an $n\times n$ matrix over a field $\field$. Let 
$
\det(\lambda I-A)=\lambda^n+c_1\lambda^{n-1}+\cdots+c_{n-1}\lambda+c_n
$
be its characteristic polynomial. The determinant of $A$ is $(-1)^nc_n$ and, if $c_n\neq 0$, its inverse is
\[
A^{-1}=-\frac{A^{n-1}+c_1A^{n-2}+\cdots+c_{n-2}A+c_{n-1}I}{c_n}.
\]
Define the vector $\vec{c}=(c_1,\ldots,c_n)^T\in\field^{n\times 1}$.
For any $k\in[n]$ let $s_k$ denote the trace of the matrix~$A^k$, and 
define the vector $\vec{s}=(s_1,\ldots,s_n)^T\in\field^{n\times 1}$.
Define the $n\times n$ matrix
\[
S=\left(
\begin{array}{ccccc}
1&&&&\\
s_1&2&&&\\
s_2&s_1&3&&\\
\vdots&\vdots&\vdots&\ddots&\\
s_{n-1}&s_{n-2}&s_{n-3}&...\:s_1&n
\end{array}
\right).
\]
It can be easily shown (see, e.g.,  \cite{Csanky+FOCS75,Prerarata+IPL78}) that 
$S\vec{c}=-\vec{s}$,
which enables us to recover $\vec{c}$ from $\vec{s}$ if $S$ is invertible.  The matrix $S$ is invertible whenever $n!\neq 0$, which is true in any field of characteristic zero or in any finite field of characteristic strictly larger than $n$. The following proposition shows that the inverse of an invertible triangular matrix can be computed efficiently in the congested clique model.  
\begin{proposition}\label{prop:trig-inv}
Let $\field$ be any field. The deterministic round complexity of $\INV{n,\field}$, when the input $A$ is an invertible lower triangular matrix,  is $O(n^{1-2/\omega})$. 
\end{proposition}
\begin{proof}
We adapt the standard sequential algorithm for triangular matrix inversion \cite{Aho+74,Bunch+74}. Let $A$ be an invertible lower triangular $n\times n$ matrix with entries in $\field$. Assume without loss of generality that $n$ is a power of two. Let us decompose $A$ in four blocks of size $n/2\times n/2$:
\[
A=
\left(
\begin{array}{cc}
A_{11}&0\\
A_{21}&A_{22}
\end{array}
\right).
\]
Observe that both $A_{11}$ and $A_{22}$ are invertible lower triangular matrices, and
\[
A^{-1}=
\left(
\begin{array}{cc}
A_{11}^{-1}&0\\
-A_{22}^{-1}A_{21}A_{11}^{-1}&A_{22}^{-1}
\end{array}
\right).
\]
Inverting $A$ thus reduces to inverting two invertible lower triangular matrices of size $n/2\times n/2$ and performing two matrix multiplications. We implement this algorithm recursively (and in parallel) in the congested clique model as follows. The $n$ nodes are partitioned into two groups of size of $n/2$: the first group consisting of nodes $1,\ldots, n/2$ and the second group consisting of nodes $n/2+1,\ldots, n$. The first group recursively computes $A_{11}^{-1}$ and the second group recursively computes $A_{22}^{-1}$. The important point here is that the computation of $A_{11}^{-1}$ and the computation of $A_{22}^{-1}$ can be done independently (i.e., in parallel). The nodes of the second group then in two rounds (using Lemma~\ref{lemma:Dolev}) distribute appropriately the rows of $A_{22}^{-1}$ to the nodes of the first group (who own the columns of~$A_{21}$), so that the nodes of the first group can compute $A_{22}^{-1}A_{21}$ and then $-A_{22}^{-1}A_{21}A_{11}^{-1}$. The nodes of the first group finally distribute appropriately the rows of $-A_{22}^{-1}A_{21}A_{11}^{-1}$ to the nodes of the second group, in two rounds from Lemma \ref{lemma:Dolev}.

Let $R_I(n)$ denote the round complexity of our problem (computing the inverse of an invertible lower triangular $n\times n$ matrix using $n$ nodes), and $R_M(n)$ denote the round complexity of computing the product of two $n\times n$ matrices using $n$ nodes (i.e., the problem $\MM{n,n,1,\field}$). 
The recurrence relation we obtain is 
\[
\left\{
\begin{array}{ll}
R_I(1)=0,&\\
R_I(n)\le R_I(n/2) + 2R_M(n/2)+4 &\textrm{ for }\:\:n\ge 2,
\end{array}
\right.
\]
which gives $R_I(n)=O(n^{1-2/\omega})$ since $R_M(n)=O(n^{1-2/\omega})$ from Theorem \ref{theorem:main}.
\end{proof}

We are now ready to give the proof of Theorem \ref{theorem:det-inv}.
\begin{proof}[Proof of Theorem \ref{theorem:det-inv}]
For convenience we assume that $n$ is a square, and write $p=\sqrt{n}$. If $n$ is not a square we can easily adapt the proof by taking $p=\ceil{\sqrt{n}}$. Observe that any integer $a\in\{0,1,\ldots,n-1\}$ can be written in a unique way as $a=(a_1-1)p+(a_2-1)$ with $a_1,a_2\in [p]$. Below when we write $a=(a_1,a_2)\in[n]$, we mean that $a_1$ and $a_2$ are the two elements in $[p]$ such that $a=(a_1-1)p+(a_2-1)$.

For any $\ell\in[n]$, let $R_\ell$ be the $p\times n$ matrix such that the $i$-th row of $R_\ell$ is the $\ell$-th row of $A^{(i-1)p}$, for each $i\in[p]$. Similarly, for any $\ell\in[n]$, let $C_\ell$ be the $n\times p$ matrix such that the $j$-th column of $C_\ell$ is the $\ell$-th column of $A^{j-1}$, for each $j\in[p]$. For each $\ell\in[n]$ define $U_\ell=R_\ell C_\ell$, which is a $p\times p$ matrix. Observe that, for any $k=(k_1,k_2)\in[n]$, the identity
\begin{equation}\label{eq:eqtrace}
s_k=\sum_{\ell=1}^n U_\ell[k_1,k_2]
\end{equation}  
holds. We will use this expression, together with the equation $\vec{c}=-S^{-1}\vec{s}$ to compute the determinant in the congested clique model.

In order to compute the inverse of $A$ we then use the following approach.
For any $(a_1,a_2)\in[p]\times[p]$, define the coefficient
$c_{a_1,a_2}\in\field$ as follows:
\[
c_{a_1,a_2}=\left\{
\begin{array}{ll}
c_{n-1-(a_1-1)p-(a_2-1)}&\textrm{ if }(a_1,a_2)\neq (p,p),\\
1&\textrm{ if }(a_1,a_2)=(p,p).
\end{array}
\right.
\]
For any $a_2\in[p]$, define the $n\times n$ matrix $E_{a_2}$ as follows:
\[
E_{a_2}=\sum_{a_1=1}^p c_{a_1,a_2}A^{(a_1-1)p}.
\]
Note that the following holds whenever $c_n\neq 0$:
\begin{equation}\label{eq:form}
A^{-1}=-\frac{\sum_{a=0}^{n-1}c_{n-1-a}A^a}{c_n}
=-\frac{\sum_{a_1=1}^{p}\sum_{a_2=1}^p c_{a_2,a_1}A^{(a_1-1)p+(a_2-1)}}{c_n}
=-\frac{\sum_{a_2=1}^{p} E_{a_2}A^{a_2-1}}{c_n}.
\end{equation}

The algorithm for $\DET{n,\field}$ and $\INV{n,\field}$ is described in Figure~\ref{fig:inverse}. 
Steps 1 and~7.2 can be implemented in $O(p^{2/\omega}n^{1-2/\omega})$ rounds from Theorem~\ref{theorem:main} (or its simplified version in the introduction). Step 5 can be implemented in $O(n^{1-2/\omega})$ rounds, again from Theorem~\ref{theorem:main}. At Steps 2, 3 and 6 each node receives $n$ elements from the field~$\field$, so each of these three steps can be implemented in two rounds from Lemma \ref{lemma:Dolev}. The other steps (Steps~4, 7.1 and~7.3) do not require any communication. 
The total round complexity of the algorithm is thus 
$
O\left(p^{2/\omega}n^{1-2/\omega}\right)=O\left(n^{1-1/\omega}\right),
$ 
as claimed.
\begin{figure}[ht!]
\begin{center}
\fbox{
\begin{minipage}{15.5 cm} 
\begin{itemize}
\item[1.]%($O(p^{2/\omega}n^{1-2/\omega})$ rounds)
The matrices $A^{(a_1-1)p}$ and $A^{a_2-1}$ are computed for all $a_1,a_2\in[p]$ using the distributed algorithm of Theorem \ref{theorem:main}. At the end of this step node $\ell\in[n]$ has the whole $p\times n$ matrix $R_\ell$ and the whole $n\times p$ matrix $C_\ell$. 
\item[2.]%(1 round)
Node~$\ell\in[n]$ locally computes $U_\ell$, and sends $U_\ell[k_1,k_2]$ to each node $k=(k_1,k_2)\in[n]$. 
\item[3.]%(1 round)
Node $k=(k_1,k_2)\in[n]$, who received $U_\ell[k_1,k_2]$ for all $\ell\in[n]$ at the previous step, locally computes $s_k$ using Equation (\ref{eq:eqtrace}). Node $k$ then sends $s_k$ to all the nodes.
\item[4.]%(0 round)
Node $\ell\in[n]$, who received $\vec{s}$ at Step 3, locally constructs $S[\ell,\ast]$ and $S[\ast,\ell]$.
\item[5.]%($O(n^{1-2/\omega})$ rounds)
The matrix $S^{-1}$ is computed using the algorithm of Proposition~\ref{prop:trig-inv}.
At the end of this step,
node $\ell\in[n]$ has $S^{-1}[\ell,\ast]$ and $S^{-1}[\ast,\ell]$. 
\item[6.]%(1 round)
Node $\ell\in[n]$ locally computes $c_\ell$ from $S^{-1}[\ell,\ast]$ and $\vec{s}$,  and sends $c_\ell$ to all nodes.
\item[7.]
The determinant of $A$ is $(-1)^nc_n$.
If $c_n=0$ the matrix $A$ is not invertible. Otherwise the nodes compute $A^{-1}$ as follows:
\begin{itemize}
\item[7.1]%(0 round)
Node $\ell\in[n]$ computes $E_{a_2}[\ell,\ast]$ for each $a_2\in[p]$ (this can be done locally since $\vec{c}$ and each row $A^{(a_1-1)p}[\ell,\ast]$ are known from Steps 6 and 1, respectively). 
\item[7.2]%($O(p^{2/\omega}n^{1-2/\omega})$ rounds)
The matrices $E_{a_2}A^{a_2-1}$ are computed for all $a_2\in[p]$ using the distributed algorithm of Theorem \ref{theorem:main} (since, besides $E_{a_2}[\ell,\ast]$ obtained at the previous step, each node $\ell\in[n]$ knows $A^{a_2-1}[\ast,\ell]$ from the result of the computation of Step 1). At the end of this step, node $\ell\in[n]$ has the $\ell$-th row and the $\ell$-th column of the matrix $E_{a_2}A^{a_2-1}$ for all $a_2\in[p]$. 
\item[7.3]%(0 round)
Node $\ell\in[n]$ computes locally $A^{-1}[\ell,\ast]$ and $A^{-1}[\ast,\ell]$ using Equation (\ref{eq:form}). 
\end{itemize}
\end{itemize}
\end{minipage}
}
\end{center}\vspace{-4mm}
\caption{Distributed algorithm for computing the determinant of an $n\times n$ matrix $A$ and computing~$A^{-1}$ if $\det(A)\neq 0$. Initially each node $\ell\in[n]$ has as input $A[\ell,\ast]$ and $A[\ast,\ell]$.}\label{fig:inverse}
\end{figure}
\end{proof}

\section{Randomized Algorithms for Algebraic Problems}
In this section we present $O(n^{1-2/\omega})$-round randomized algorithms for computing the determinant, the rank and for solving linear systems of equations in the congested clique model, and prove Theorem \ref{th:rand}. Our approach is based on Wiedemann's method \cite{Wiedemann86} and its parallel implementations \cite{Kaltofen+SPAA91,Kaltofen+FOCS92,Kaltofen+91}. 

Let $A$ be an $n\times n$ matrix over $\field$. The minimal polynomial of $A$, which we denote $\minpoly(A)$, is the monic polynomial $g$ over $\field$ of least degree such that $g(A)=0$. Let $v\in\field^{n\times 1}$ and $w\in \field^{1\times n}$ be any vectors. Consider the sequence $wA^0v,wAv,wA^2v,\ldots,wA^{n-1}v$ consisting of~$n$ elements of $\field$. This sequence is linearly generated over $\field$, and thus also admits a polynomial called the generating polynomial of the sequence, which we denote $\minpoly(A,v,w)$. We refer to \cite{Gathen+03} for the precise definition of the generating polynomial of such a linearly generated sequence and a description of efficient algorithms to compute it --- in this paper we will just need to know that such a polynomial exists. Wiedemann \cite{Wiedemann86} showed that with high probability $\minpoly(A,v,w)=\minpoly(A)$ when~$v$ and $w$ are chosen at random. We use the following characterization of this property proved by Kaltofen and Pan \cite{Kaltofen+FOCS92}.  

\begin{lemma}(\cite{Kaltofen+FOCS92})\label{lemma:kal}
Let $A$ be any $n\times n$ matrix over $\field$.
Let $v\in\field^{n\times 1}$ and $w\in \field^{1\times n}$ be two vectors in which each coordinate is chosen uniformly at random from $\field$. Then
\[
\Pr\big[\minpoly(A)=\minpoly(A,v,w)\big]\ge 1-\frac{2n}{|\field|}.
\]
\end{lemma}

In general the minimal polynomial of a matrix $A$ is not equal to its characteristic polynomial $\charpoly(A)$. Wiedemann \cite{Wiedemann86} nevertheless showed that the characteristic polynomial can be obtained from the minimal polynomial by preconditioning the matrix. Since it will be more convenient for our purpose to apply a diagonal preconditioner, we will use the following version shown by Chen et al.~\cite{Chen+02}.

\begin{lemma}(\cite{Chen+02})\label{lemma:wie}
Let $A$ be any $n\times n$ matrix over $\field$.
Let $D$ be an $n\times n$ diagonal matrix where each diagonal entry is chosen uniformly at random from $\field\setminus\{0\}$.
Then
\[
\Pr\big[\charpoly(DA)=\minpoly(DA)\big]\ge 1-\frac{n(n-1)}{2(|\field|-1)}.
\]
\end{lemma}

Kaltofen and Saunders \cite{Kaltofen+91} showed that the rank of a matrix can be obtained, with high probability, from the degree of the minimal polynomial by using another preconditioning of the matrix. The following lemma follows from the combination of Theorem 1 and Lemma 2 in \cite{Kaltofen+91}.  

\begin{lemma}(\cite{Kaltofen+91})\label{lemma:rank}
Let $A$ be any $n\times n$ matrix over $\field$ such that $\rk(A)<n$.
Let 
\[
U=\left(
\begin{array}{ccccc}
1&u_2&u_3&\cdots&u_n\\
&1&u_2&\cdots&u_{n-1}\\
&&1&\ddots&\vdots\\
&&&\ddots&u_2\\
&&&&1
\end{array}
\right),\hspace{5mm}
V=\left(
\begin{array}{ccccc}
1&&&&\\
v_2&1&&&\\
v_3&v_2&1&&\\
\vdots&&\ddots&\ddots&\\
v_n&v_{n-1}&\cdots&v_2&1
\end{array}
\right)
\]
be two $n\times n$ unit (upper and lower, respectively) triangular random Toepliz matrices. Let $D$ be an $n\times n$ diagonal matrix where each diagonal entry is chosen uniformly at random from $\field$.
Then
\[
\Pr\big[\rk(A)=\deg(\minpoly(UAVD))-1\big]\ge 1-\frac{n(3n+1)}{2|\field|}.
\]
\end{lemma}

We are now ready to prove Theorem \ref{th:rand}.

\begin{proof}[Proof of Theorem \ref{th:rand}]
We first show how to compute efficiently, in the congested clique model, the sequence $\tilde A^0u,\tilde Au,\tilde A^2u,\ldots$,$\tilde A^{n-1}u$ given an arbitrary  matrix $\tilde A\in\field^{n\times n}$ and an arbitrary vector $u\in\field^{n\times 1}$. For convenience assume that $n$ is a power of two (otherwise we can simply add zero rows and columns to $\tilde A$ and zero entries to $u$), and write $n=2^k$ for some positive integer $k$.
For each $i\in\{0,\ldots,k\}$, define the $n\times n$ matrix 
\begin{align*}
M^{(i)}&=[u|\tilde Au|\cdots|\tilde A^{2^i-1}u|0|\cdots|0],
\end{align*}
obtained by concatenating the vectors $u,\ldots, \tilde A^{2^i-1}u$ and then adding $2^{k}-2^i$ zero columns. For each $i\in\{0,\ldots,k-1\}$, define the $n\times n$ matrix 
\begin{align*}
N^{(i)}&=[0|\cdots|0|u|\tilde Au|\cdots|\tilde A^{2^i-1}u|0|\cdots|0],
\end{align*}
obtained by concatenating the vectors $u,\ldots, \tilde A^{2^i-1}u$, adding $2^i$ zero columns on the left and $2^{k}-2^{i+1}$ zero columns on the right.
Observe that, for any $i\in\{0,\ldots,k-1\}$, the equality
\begin{align*}
M^{(i+1)}&=M^{(i)} + \tilde A^{2^i}N^{(i)}
\end{align*}
holds. Moreover, for any $i\in\{1,\ldots,k-1\}$, the matrix $\tilde A^{2^{i}}$ can be obtained by multiplying $\tilde A^{2^{i-1}}$ by itself. This enables us to compute the $n$ vectors $\tilde A^0u,\tilde Au,\tilde A^2u,\ldots,\tilde A^{n-1}u$ using only $2k-1$ matrix multiplications. We describe in Figure~\ref{fig:rand} the implementation of this approach in the congested clique model, which uses $O(kn^{1-2/\omega})=O(n^{1-2/\omega}\log n)$ rounds.  

\begin{figure}[ht!]
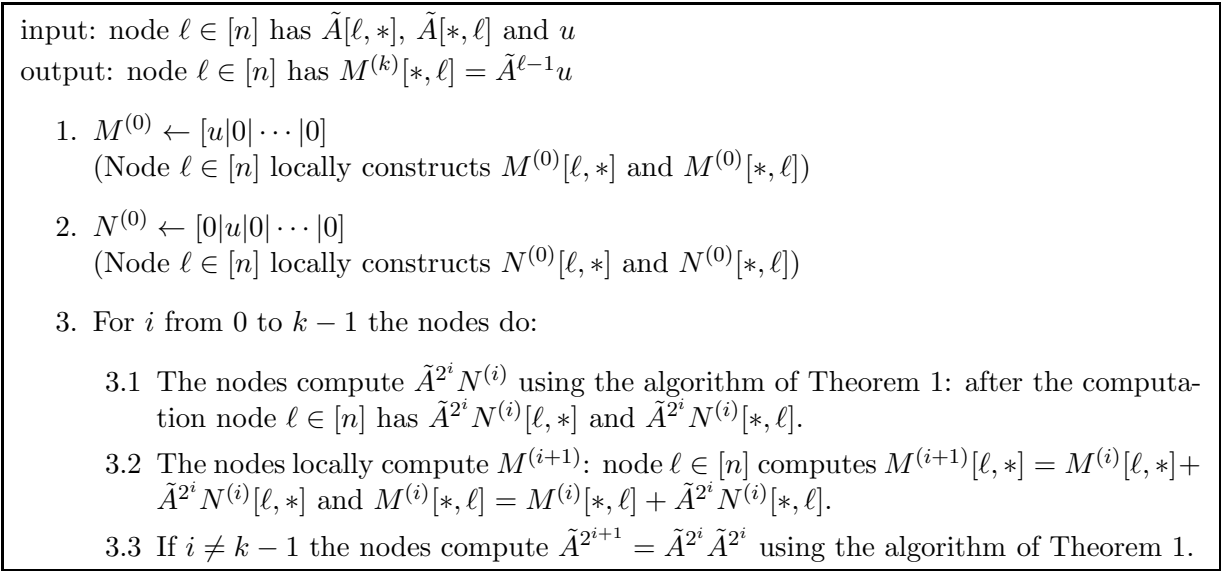

\begin{center}
\fbox{
\begin{minipage}{15.7 cm} 
input: node $\ell\in[n]$ has $\tilde A[\ell,\ast]$, $\tilde A[\ast,\ell]$ and $u$\\\vspace{-4mm}

\noindent 
output: node $\ell\in[n]$ has $M^{(k)}[\ast,\ell]=\tilde A^{\ell-1}u$
\begin{itemize}
\item[1.]
$M^{(0)}\gets [u|0|\cdots|0]$\\
(Node $\ell\in[n]$ locally constructs $M^{(0)}[\ell,\ast]$ and $M^{(0)}[\ast,\ell]$)
\item[2.]
$N^{(0)}\gets [0|u|0|\cdots|0]$\\
(Node $\ell\in[n]$ locally constructs $N^{(0)}[\ell,\ast]$ and $N^{(0)}[\ast,\ell]$)
\item[3.]
For $i$ from 0 to $k-1$ the nodes do:
\begin{itemize}
\item[3.1]
The nodes compute $\tilde A^{2^i}N^{(i)}$ using the algorithm of Theorem \ref{theorem:main}: after the computation node $\ell\in[n]$ has $\tilde A^{2^i}N^{(i)}[\ell,\ast]$ and $\tilde A^{2^i}N^{(i)}[\ast,\ell]$.
\item[3.2]
The nodes locally compute $M^{(i+1)}$: node $\ell\in[n]$ computes $M^{(i+1)}[\ell,\ast]=M^{(i)}[\ell,\ast]+\tilde A^{2^i}N^{(i)}[\ell,\ast]$ and $M^{(i)}[\ast,\ell]=M^{(i)}[\ast,\ell]+\tilde A^{2^i}N^{(i)}[\ast,\ell]$.
\item[3.3]
If $i\neq k-1$ the nodes compute $\tilde A^{2^{i+1}}=\tilde A^{2^{i}}\tilde A^{2^{i}}$ using the algorithm of Theorem~\ref{theorem:main}.
\end{itemize}
\end{itemize}
\end{minipage}
}
\end{center}\vspace{-4mm}
\caption{Distributed randomized algorithm for computing the sequence 
$\tilde A^0u,\tilde Au,\tilde A^2u,\ldots,\tilde A^{n-1}u$ given a matrix $\tilde A\in\field^{n\times n}$ and a vector $u\in\field^{n\times 1}$. Here $n$ is a power of two, written $n=2^k$.}\label{fig:rand}
\end{figure}

We now describe a $O(n^{1-2/\omega}\log n)$-round algorithm in the congested clique model that computes, with high probability, the minimal polynomial $\minpoly(A)$ of an arbitrary matrix $A\in\field^{n\times n}$. More precisely, the matrix~$A$ is initially distributed among the $n$ nodes of the network (node $\ell\in[n]$ receives as input $A[\ell,\ast]$ and $A[\ast,\ell]$) and at the end of the computation we would like a designated node of the network (say, node 1) to have the polynomial $\minpoly(A)$. The algorithm is as follows. First, a designated node (say, node 1 again) takes two random vectors $v,w$ as in Lemma \ref{lemma:kal}. This node then sends $v,w$ to all the nodes of the network in four rounds of communication using the scheme of Lemma~\ref{lemma:Dolev}. Then the nodes of the network apply the algorithm of Figure~\ref{fig:rand} with $\tilde A=A$ and $u=v$. After this, each node $\ell\in[n]$ owns $A^{\ell-1}v$, and can then compute locally the field element $wA^{\ell-1}v$. All nodes then sends their result to node~1 using one round of communication. Node 1 finally computes locally the polynomial $\minpoly(A,v,w)$. The complexity of this algorithm is clearly $O(n^{1-2/\omega}\log n)$ rounds. The correctness follows from Lemma \ref{lemma:kal}, which guarantees that $\minpoly(A)=\minpoly(A,v,w)$ with probability at least $1-2n/|\field|$. 

Our $O(n^{1-2/\omega}\log n)$-round algorithm for $\INV{n,\field}$ (i.e., for solving the linear system $Ax=b$ where $A$ is invertible) in the congested clique model is as follows.
The nodes of the network first use the algorithm of the previous paragraph, so that node 1 obtains with high probability $\minpoly(A)$. Let us write $\minpoly(A)$ as $m_0+m_1\lambda+\cdots +m_{n}\lambda^n$, with $m_n=1$. From the definition of the minimal polynomial, we have
\begin{equation}\label{eq:sys}
x=\frac{m_1b+m_2Ab+\cdots +m_{n}A^{n-1}b}{-m_0}.
\end{equation}
Node 1 then sends the two field elements $m_0$ and $m_{\ell}$ to node $\ell$, for each $\ell\in\{1,\ldots,n\}$, in two rounds. The nodes of the network use the $O(n^{1-2/\omega}\log n)$-round algorithm of Figure~\ref{fig:rand} with $u=b$ and $\tilde A=A$, so that each node $\ell\in[n]$ owns $A^{\ell-1}b$ at the end of the computation. Each node $\ell\in[n]$ then locally computes $-\frac{m_\ell}{m_0}A^{\ell-1}b$, and sends the $\ell'$-th coordinate of this vector to node $\ell'$, for each $\ell'\in[n]$. This can be done in two rounds. Node $\ell\in[n]$ then adds the $n$ elements he receives, which gives $x[\ell]$ from Equation (\ref{eq:sys}).

We now describe our $O(n^{1-2/\omega}\log n)$-algorithm for $\DET{n,\field}$. First, a designated node (say, node 1) takes a random diagonal matrix~$D$ as in Lemma \ref{lemma:wie}. This node then sends $D$ to all the nodes of the network using two rounds of communication (using the scheme of Lemma~\ref{lemma:Dolev}). Each node $\ell\in[n]$ of the network then constructs $DA[\ell,\ast]$ and $DA[\ast,\ell]$. The nodes of the network then apply the above algorithm computing the minimal polynomial, so that Node 1 obtains with high probability $\minpoly(DA)$. Let $m_0$ denote the constant term of $\minpoly(DA)$. Note that the determinant of $DA$ is $(-1)^n m_0$, from Lemma \ref{lemma:wie}, and thus the determinant of $A$ is 
\[
\frac{(-1)^n m_0}{\prod_{i=1}^n D[i,i]}.
\] 
Node 1 sends this value all the nodes of the network in one round of communication. 

Finally, we present our algorithm for $\Rank{n,\field}$. First, we can check with high probability whether $\rk(A)=n$ by computing $\det(A)$ using the algorithm described in the previous paragraph. Therefore we assume below that $\rk(A)<n$, and compute the rank as follows.
A designated node (say, node 1) takes three random matrices $U,V,D$ as in Lemma \ref{lemma:rank}, and sends the three matrices to all the nodes of the network in six rounds using the scheme of Lemma \ref{lemma:Dolev} (note that each matrix is described by at most $n$ coefficients). The nodes of the network then apply the $O(n^{1-2/\omega})$-round algorithm of Theorem \ref{theorem:main} three times to compute $UAVD$. They then use the $O(n^{1-2/\omega}\log n)$-round algorithm computing the minimal polynomial, so that Node 1 obtains $\minpoly(UAVD)$. Node 1 locally computes $\deg(\minpoly(UAVD))-1$ and sends this value to all the nodes of the network in one round. The correctness of this algorithm is guaranteed by Lemma \ref{lemma:rank}.
\end{proof}

%%%%%%%%%%%%%%%%%%%%%%%%%%%%%%%%%%%%%
\section{Applications to Graph-Theoretic Problems}\label{sec:appli}
In this section we consider applications of our results to graph-theoretic problems in the congested clique model. 

\subsection{The All Pair Shortest Paths problem}\label{sub:APSP}
The All-Pairs Shortest Paths problem (APSP) asks, given a weighted graph $G=(V,E)$, to compute the shortest path between $u$ and $v$ for all pairs of vertices $(u,v)\in V\times V$. For simplicity, but without significant loss of generality, we will assume that the weights are $O(\log n)$-bit integers. When the graph $G$ is undirected, the definition of the APSP requires the weights to be nonnegative (otherwise there would be negative cycles). When the graph $G$ is directed negative weights are allowed but it is implicitly required that the graph has no non-negative cycle. We say that the graph is unweighted if the only allowed weight is one. In this subsection the term ``adjacency matrix of $G$'' refers to the $|V|\times |V|$ matrix in which the entry in the $i$-th row and $j$-th column is the weight of the edge from the $i$-th node to the $j$-th node of the graph if these two nodes are connected and $\infty$ otherwise. 

Let us first describe very briefly the main results concerning the complexity of the APSP in the centralized setting. For undirected unweighted graphs, Seidel \cite{SeidelJCSS95} showed that the APSP reduces to computing the powers of the adjacency matrix (seen as a matrix over the integers) of the graph, and can thus be solved in $\tilde O(n^\omega)$ time. For all other cases, including directed graphs and undirected graphs with arbitrary weights, the standard algebraic way of solving the APSP is to compute the powers of the distance product of the adjacency matrix of the graph. This distance product of an $n\times m$ matrix $A$ by an $m\times n$ matrix $B$ can be trivially solved in time $O(smn^2)$, where $s$ denotes the number of bits needed to represent each entry of $A$ and $B$, which gives a $\tilde O(n^3)$-time algorithm for the APSP.
%, assuming that the weights of the graph can be expressed using $O(\log n)$ bits. 
Despite much research (including recent exciting developments \cite{WilliamsSTOC14}), no significantly better algorithm is known for computing the distance product or solving the general APSP. Faster algorithms for computing the distance product can be nevertheless designed when the entries of $A$ and $B$ are small integers, i.e., the entries are in $\{-M,\ldots,-1,0,1,\ldots,M\}\cup\{\infty\}$ for some integer~$M$. As already mentioned in the proof of Theorem \ref{th:dist}, Alon et al. \cite{Alon+97} and Takaoka~\cite{Takaoka98} (for the square case) and then Zwick \cite{ZwickJACM02} (for the rectangular case) showed that $C$ can be recovered easily from the standard matrix product of the $n\times m$ matrix $A'$ by the $m\times n$ matrix $B'$ where
\begin{equation*}%\label{eq:dist}
A'[i,j]=\left\{\!\!
\begin{array}{cl}
(m+1)^{M-A[i,j]}&\textrm{if } A[i,j]\neq\infty,\\
0&\textrm{if } A[i,j]=\infty,
\end{array}
\right.
\hspace{5mm}
B'[j,i]=\left\{\!\!
\begin{array}{cl}
(m+1)^{M-B[j,i]}&\textrm{if } B[j,i]\neq\infty,\\
0&\textrm{if } B[j,i]=\infty,
\end{array}
\right.
\end{equation*}
for all $(i,j)\in [n]\times [m]$. This implies in particular that the distance product can be computed in $\tilde O(M n^{\omega(\log m/\log n)})$ time, and in particular in $\tilde O(M n^{\omega})$ time when $m=n$. Shoshan and Zwick~\cite{Shoshan+FOCS99} used this technique (for the square case) to obtain a $\tilde O(M n^{\omega})$-time algorithm for the APSP in undirected graphs with weights in $\{0,\ldots,M\}$. Zwick used this technique (for the rectangular case) to construct an algorithm for the directed case. He obtained in particular time complexity $O(n^{2.58})$, which has been improved to $O(n^{2.54})$ using the best known upper bound on the exponent of rectangular matrix multiplication \cite{LeGallFOCS12}, when $M$ is constant (and in particular for directed unweighted graphs).

Censor-Hillel et al.~\cite{Censor-Hillel+15} showed how to adapt Seidel's method \cite{SeidelJCSS95} to solve the APSP over undirected unweighted graphs in $\tilde O(n^{1-2/\omega})$ rounds.  They also observed that the centralized matrix algorithms for the distance product discussed above can be implemented in the congested clique model, solving $\DIS{n,n,M}$ in $O(\min\{n^{1/3}\log n,Mn^{1-2/\omega}\})$ rounds. While not explicitly stated in \cite{Censor-Hillel+15}, this result gives a $O(Mn^{1-2/\omega})$-round algorithm for the APSP in undirected graphs with weights in $\{0,1,\ldots,M\}$, by observing that the reduction given in \cite{Shoshan+FOCS99}  from such instances of APSP to the computation of the distance product can be implemented efficiently in the congested clique model. Our improved upper bound (Theorem \ref{th:APSPu} stated in the introduction) follows directly from our improved algorithm for the computation of the distance product in the congested clique model (Theorem \ref{th:dist}). 

We now consider the APSP in directed graphs with small integer weights and prove Theorem~\ref{th:APSPd}. Besides the new bounds of Theorem \ref{th:dist}, we will also use another upper bound on the round complexity of $\DIS{n,m,M}$, which is better for large values of $M$. This is the bound obtained from Theorem \ref{theorem:main} for the choice $\omega(\ell)=2+\ell$ corresponding to the implementation of the trivial matrix multiplication algorithm (which works over any semiring, see Section \ref{sec:prelim}) in the congested clique model. We state this bound in the following proposition.

\begin{proposition}\label{prop:dist}
The round complexity of $\DIS{n,m,M}$ is
\[
\left\{
\begin{array}{ll}
O(\log M)&\textrm{ if }\:\: 0\le m\le \sqrt{n},\\
O(m^{2/3}n^{-1/3}\log M) &\textrm{ if }\:\:m\ge \sqrt{n}.
\end{array}
\right.
\]
\end{proposition}
We are now ready to prove Theorem \ref{th:APSPd}.
\begin{proof}[Proof of Theorem \ref{th:APSPd}]
Let $A$ be the adjacency matrix of the graph $G=(V,E)$. The centralized algorithm by Zwick~\cite{ZwickJACM02} works as follows. The algorithm performs $\lceil\log_{3/2}n\rceil$ iteration, while maintaining an $n\times n$ matrix $F$ initially set to $F=A$. In the $k$-th iteration, it sets $s=(3/2)^k$, and takes a set $S\subseteq V$ of $O((n\log n)/s)$ vertices chosen uniformly at random from $V$. The algorithm then construct the submatrix of size $|S|\times n$ of $F$, denoted $F[S,\ast]$, consisting of the rows corresponding to the vertices in~$S$, and the submatrix of size $n\times |S|$ of $F$, denoted $F[\ast,S]$, consisting of the columns corresponding to vertices in $S$. It then puts a cap of $sM$ on the absolute values of the entries of $F[S,\ast]$ and $F[\ast,S]$. The last step of the iteration is to compute the distance product $F'=F[\ast,S]\dist F[S,\ast]$ and, for each $(i,j)\in[n\times n]$, replace the $(i,j)$ entry of $F$ by $F'[i,j]$ if $F'[i,j]<F[i,j]$. It can be shown that after the last iteration the entry $F[i,j]$ is the length of the shortest path between the $i$-th node and the $j$-th node of the graph, for all $(i,j)\in[n\times n]$, if the graph has no negative cycle.

This centralized algorithm can be implemented easily in the congested clique model. The only part that requires communication between the nodes is the computation of the distance product: at step $i$ the nodes need to compute the distance product of a $n\times m$ matrix by a $m\times n$ matrix with entries of absolute values bounded by $\ceil{sM}$ with $m=O((n\log n)/s)$ and $s=(3/2)^i$. We have two strategies to compute the distance product, the algorithm of Proposition \ref{prop:dist} and the algorithm of Theorem \ref{th:dist}. Since there are only $O(\log n)$ iterations, the total round complexity is
\begin{equation}\label{eq:opt}
\tilde O\left(
\min\left\{1+(n/s)^{2/3}n^{-1/3},
s+s^{2/\omega(\gamma)}n^{1-2/\omega(\gamma)}
\right\}\right),
\end{equation}
for the value of $s$ the maximizes this expression,
where $\gamma$ denotes the solution of the equation $
\big(1-\frac{\log s}{\log n}\big)\gamma=1-\frac{\log s}{\log n}+\big(\frac{\log (n/s)}{\log n}-1\big)\omega(\gamma).
$
As in the analysis for the centralized setting given in~\cite{ZwickJACM02}, the left part of (\ref{eq:opt}) is a decreasing function of $s$, while the right part is an increasing function of $s$. Using the best known upper bound on $\omega(\gamma)$ from \cite{LeGallFOCS12} (see also Figure \ref{fig:opt}), we can upper bound the total round complexity by $O(n^{0.2096})$.

Note that the above algorithm only computes the lengths of the shortest paths. As in Zwick's centralized algorithm, the shortest paths can be constructed by using exactly the same strategy, but constructing a matrix of witnesses whenever a distance product is computed (as described in Section 3 of \cite{ZwickJACM02}), which can be done with the same round complexity.  
\end{proof}
While a result similar to Theorem \ref{th:APSPd} can be obtained when $M$ is not a constant, by using exactly the same algorithm but keeping the value $M$ in Equation (\ref{eq:opt}), in this case it is complicated to express the result of the numerical optimization in a closed form, so we omit this generalization. An interesting open question is whether the algorithm of Theorem~\ref{th:APSPd} can be derandomized. While Zwick showed that this 
can be done in the centralized setting by introducing the concept of bridging sets, it does not seem that the algorithm proposed in \cite{ZwickJACM02} for constructing bridging sets can be implemented efficiently in the congested clique model.

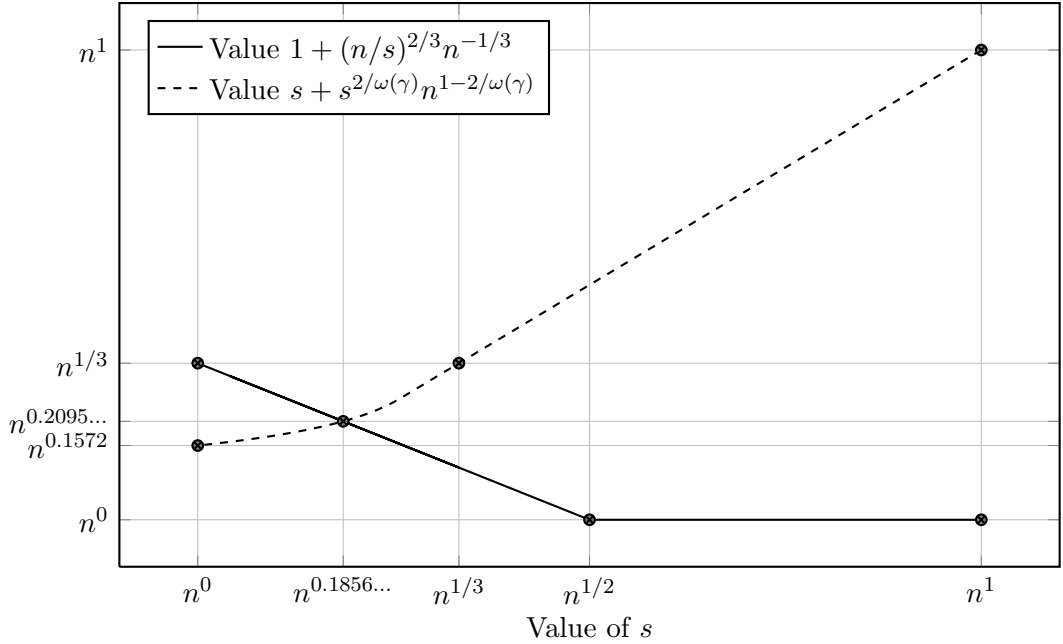
\begin{figure}[ht]
\centering
\begin{tikzpicture}
\begin{axis}[
legend cell align=left,
width=12.5cm, 
height=7.5cm,
xmin=-0.1, xmax=1.1,
thick,
 scale only axis,
xmajorgrids,
ymajorgrids,
 xtick={0,0.185677,0.33333,0.5,1},
 xticklabels={$n^0$,$n^{0.1856\ldots}$,$n^{1/3}$,$n^{1/2}$,$n^1$},
 ytick={0,0.15806,0.209548,0.33333,1},
 yticklabels={$n^0$,$n^{0.1572}$,$n^{0.2095\ldots}$,$n^{1/3}$,$n^1$,$n^2$},
xlabel={Value of $s$},
ylabel={},
legend pos=north west]

\addplot [solid, every mark/.append style={solid, fill=gray}] coordinates {
(0,0.33333)
(0.333333333,0.111111111)
(0.258440798,0.161039468)
(0.256499569,0.162333621)
(0.253696223,0.164202518)
(0.25084811,0.16610126)
(0.247955215,0.168029857)
(0.24501766,0.169988227)
(0.229693646,0.180204236)
(0.213418018,0.191054655)
(0.196333974,0.202444018)
(0.185677315,0.209548457)
(0.178566964,0.21428869)
(0.160219757,0.226520162)
(0.141375527,0.239082982)
(0.12210216,0.251931894)
(0.102455363,0.265029758)
(0.08248137,0.278345754)
(0.062219236,0.291853843)
(0.041700583,0.305532945)
(0.5,0)
(1,0)
};

\addplot [dashed, every mark/.append style={solid, fill=gray}] coordinates {
(1,1)
(0.333333333,0.333333333)
(0.258440798,0.258440798)
(0.256499569,0.256522988)
(0.253696223,0.253834637)
(0.25084811,0.251199672)
(0.247955215,0.248620561)
(0.24501766,0.246099507)
(0.229693646,0.234354513)
(0.213418018,0.223934481)
(0.196333974,0.214664106)
(0.185677315,0.209547403)
(0.178566964,0.206369047)
(0.160219757,0.198901213)
(0.141375527,0.192139844)
(0.12210216,0.185985603)
(0.102455363,0.180357094)
(0.08248137,0.175186305)
(0.062219236,0.170410191)
(0.041700583,0.165988337)
(0.020952887,0.161884519)
(0.0,0.158063833)
};

\addplot [only marks, every mark/.append style={solid, fill=gray}, mark=otimes*] coordinates {
(0,0.1580)
(0,0.333333)
(0.185677,0.209548)
(0.5,0)
(0.333333,0.333333)
%(1,0.0)
(1,0)
(1,1)
(2,1)
(3,2)
};
\legend{Value $1+(n/s)^{2/3}n^{-1/3}$, Value $s+s^{2/\omega(\gamma)}n^{1-2/\omega(\gamma)}$}
%, Ideal algorithm}
\end{axis}
\end{tikzpicture}
\caption{\label{fig:opt}Values of the two parts of Equation (\ref{eq:opt}).}
\end{figure}

\subsection{Diameter computation}\label{sub:diam}
Since computing the diameter of a graph trivially reduces to the APSP problem, we immediately obtain the following result.
\begin{corollary}\label{cor:diameter}
In the congested clique model, the deterministic round complexity of diameter computation in an undirected graph of $n$ vertices with integer weights in $\{0,\ldots,M\}$, where~$M$ is an integer such that $M\le n$, is $\tilde O(M^{2/\omega}n^{1-2/\omega})$.
\end{corollary}

The same reduction can be used to obtain a randomized algorithm for computing the diameter of directed graphs with constant integer weights in $O(n^{0.2096})$ rounds, via Theorem \ref{th:APSPd}. This round complexity can nevertheless be improved. In the centralized setting it is known that over directed graphs with integer weights in $\{-M,\ldots,M\}$, but without cycles of negative weights, the diameter can be computed in $O(Mn^\omega)$ time, i.e., faster than the best known centralized algorithm for the corresponding APSP problem. This is a folklore result based on the reductions to the distance product developed in \cite{Shoshan+FOCS99,ZwickJACM02} (see also \cite{Cygan+FOCS12}). A close inspection of this approach shows that it can be implemented efficiently in the congested clique model, and Theorem \ref{th:dist} thus implies that the diameter can be computed in $\tilde O(M^{2/\omega}n^{1-2/\omega})$ rounds in direct graphs (without cycles of negative weights) with integer weights in $\{-M,\ldots,M\}$ as well. 

\subsection{Maximum matchings}\label{sub:match}
Let $G=(V,E)$ be a simple graph (i.e., an undirected and unweighted graph with no loops or multiple edges), and write $V=\{v_1,\ldots,v_n\}$. The Tutte matrix of $G$ is the $n\times n$ symbolic matrix $A$ such that 
\[
A[i,j]=\left\{
\begin{array}{cl}
x_{ij}&\textrm{ if } i> j \textrm{ and  } \{v_i,v_j\}\in E,\\
-x_{ij}&\textrm{ if } i< j \textrm{ and  } \{v_i,v_j\}\in E,\\
0&\textrm{ otherwise},
\end{array}
\right.
\]
for any $(i,j)\in[n]\times [n]$. Let $\nu(G)$ denote the number of edges in a maximum matching of $G$. Lov\'asz~\cite{Lovasz79} showed that $\rk(A)=2\nu(G)$. Rabin and Vazirani \cite{Rabin+89} observed that this equality remains true over any field. As mentioned in \cite{Rabin+89}, this gives a simple randomized algorithm for computing $\rk(A)$, and thus $\nu(G)$, based on Schwartz-Zippel lemma \cite{Schwartz80, Zippel79}: take a prime $p=\Theta(n^4)$, substitute the variables $x_{ij}$ by random elements from the finite field $\Int_p$  to obtain a matrix $\hat A$ over $\Int_p$, and compute the rank of $\hat A$ over~$\Int_p$ (which will be equal to $\rk(A)$ with high probability). This approach can clearly be implemented efficiently in the congested clique mode using the rank algorithm of Theorem \ref{th:rand}, giving the following result.

\begin{theorem}\label{th:mm}
The randomized round complexity of computing the number of edges in a maximum matching of a simple graph is $O(n^{1-2/\omega}\log n)$.
\end{theorem}

Suppose that the graph $G$ has a perfect matching (i.e., $n$ is even and $\nu(G)=n/2)$. We say that an edge of $G$ is allowed if it is contained in a least one perfect matching. Rabin and Vazirani \cite{Rabin+89} further showed that with high probability the following property holds for all edges $\{i,j\}$ of $G$: 
the edge $\{i,j\}$ is allowed if and only if $\hat A^{-1}[i,j]\neq 0$. The set of allowed edges can thus be obtained from the inverse of the matrix $\hat A$, which can be done in $O(n^{1-1/\omega})$ rounds in the congested clique model using the algorithm of Theorem \ref{theorem:det-inv}. 

An interesting open question is whether finding a perfect matching can be done with the same complexity in the congested clique model. While the best centralized algorithms can find a maximum matching with essentially the same complexity as matrix multiplication \cite{Mucha+FOCS04}, they are based on sequential variants of the Gaussian decomposition (e.g., computation of the LUP decomposition) that do not appear to be implementable in parallel.

\subsection{Computing the Gallai-Edmonds decomposition of a graph}\label{sub:GE}
Let $G=(V,E)$ be a simple graph. We say that a vertex $v\in V$ is critical if it appears in at least one maximum matching, otherwise we say that $v$ is non-critical.
Let $D(G)\subseteq V$ denote the set of non-critical vertices, $K(G)$ be the set of vertices in $V\setminus D(G)$ that are adjacent to vertices in $D(G)$, and define $C(G)=V\setminus(D(G)\cap K(G))$. Gallai~\cite{Gallai64} and Edmonds \cite{Edmonds65} showed that the triple $(D(G),K(G),C(G))$, called the Gallai-Edmonds decomposition of the graph, gives fundamental information about the structure of the graph, and in particular about its matchings (see, e.g.,~\cite{Lovasz+09} for a detailed presentation of this theorem). %For instance, each component of the subgraph induced by $C(G)$ has a perfect matching. 
Cheriyan \cite{CheriyanSICOMP97} presented efficient algorithms for computing the Gallai-Edmonds decomposition. We adapt this algorithm to the congested clique model to obtain the following result.
\begin{theorem}\label{th:GE}
The randomized round complexity of computing the Gallai-Edmonds decomposition of a simple graph is $O(n^{1-1/\omega})$.
\end{theorem}
\begin{proof}
We will show how all the nodes of the network can obtain the set of non-critical vertices $D(G)$ in $\tilde O(n^{1-1/\omega})$ rounds. Note that $K(G)$ and $C(G)$ can then be computed (and distributed to all the nodes) in a constant number of rounds using the scheme of Lemma \ref{lemma:Dolev}.

The strategy used in \cite{CheriyanSICOMP97} to compute the set of non-critical vertices $D(G)$ in the centralized setting is as follows. Take a prime $p=\Theta(n^3)$ and substitute the variables $x_{ij}$ in the Tutte matrix of $G$ by random elements from the finite field $\field=\Int_p$. Let $\hat A$ denote the matrix obtained. For any $i\in[n]$, let $e_i\in\field^{1\times n}$ be the row vector with coordinate $1$ at position $i$ and coordinate zero at all other positions. Cheriyan~\cite{CheriyanSICOMP97} showed that with large probability the following property holds for all $i\in[n]$: the vertex $v_i$ is non-critical if and only if 
%  $i$ is critical if and only if the rank of the $(n+1)\times n$ matrix $\begin{bmatrix}B\\e_i\end{bmatrix}$ is larger than the rank of $B$. 
\begin{equation}\label{eq:rank}
\rk\left(
\left[
\begin{array}{c}
\hat A\\
e_i
\end{array}
\right]
\right)
>
\rk(\hat A).
\end{equation}
Here $\begin{bmatrix}\hat A\\e_i\end{bmatrix}$ denotes the $(n+1)\times n$ matrix obtained by appending the row $e_i$ at the bottom of~$\hat A$. The set $D(G)$ can be computed by checking if Equation (\ref{eq:rank}) holds for each $i\in[n]$, but this is not efficient enough. Observe that Equation (\ref{eq:rank}) holds if and only if $e_i$ is not in the subspace~$S$ spanned by the row vectors of $\hat A$. Let $M\in\field^{n\times (n-\rk(\hat A))}$ be the matrix representation of a basis for the right null space of $\hat A$ (i.e., each column of $M$ is a basis vector of the vector space $S^\perp=\{y\in\field^n\:|\:\hat Ay=0\}$). Observe that for any row vector $u$ we have  $uM=0$ if and only if $u\in (S^\perp)^\perp=S$.
%Since $(\hat A^\perp)^\perp=A$, 
Equation (\ref{eq:rank}) then holds if and only if $e_iM\neq 0$. In order to compute $e_iM$ for all $i\in[n]$, we simply need to compute the product of the $n\times n$ identity matrix $I_n$ by the matrix $\hat A$, and check which rows of the product contain at least one non-zero entry.

In the congested clique model this strategy can be implemented in $O(n^{1-2/\omega})$ rounds assuming that the matrix $M$ is available (and distributed among the nodes). We now explain how to construct this matrix using the ideas from \cite{Kaltofen+SPAA91,Kaltofen+91}. Let $r$ denote the rank of $A$. Let $U,V$ be two triangular random Toeplitz matrices as in Lemma \ref{lemma:rank}, and write $N=U\hat AV$. Decompose $N$ as follows:
\[
N=
\left(
\begin{array}{cc}
N_{11} & N_{12}\\
N_{21}& N_{22}
\end{array}
\right),
\]
where $N_{11}\in\field^{r\times r}$, $N_{12}\in\field^{r\times (n-r)}$,
$N_{21}\in\field^{(n-r)\times r}$ and $N_{22}\in\field^{(n-r)\times (n-r)}$. Theorem 2 from \cite{Kaltofen+91} and the analysis of Section 5 in \cite{Kaltofen+SPAA91} shows that, with probability at least $1-r(r+1)/|\field|$ on the choice of $U$ and $V$, the submatrix $N_{11}$ is non-singular and the columns of the $n\times (n-r)$ matrix
\[
M=V
\left(
\begin{array}{cc}
I_r &-N_{11}^{-1} N_{12}\\
0^{(n-r)\times r}& I_{n-r}
\end{array}
\right)
\left(
\begin{array}{c}
0^{r\times (n-r)}\\
I_{n-r}
\end{array}
\right).
\]
form a basis of the right null space of $\hat A$. The matrix $M$ can thus be computed with high probability in $O(n^{1-1/\omega})$ rounds by using the algorithm of Theorem \ref{theorem:det-inv} for computing the inverse and the algorithm for matrix multiplication of Section \ref{sec:mm}.
%Such a kernel can be found using the ideas from \cite{Kaltofen+FOCS92}.
\end{proof}

As already mentioned, the Gallai-Edmonds decomposition has many applications. In particular, as pointed out by Cheriyan \cite{CheriyanSICOMP97},  an algorithm computing the Gallai-Edmonds decomposition immediately yields an algorithm computing a minimum vertex cover in a bipartite graph. 
Cheriyan also presented other graph-theoretic problems that can be solved using variants of his approach for computing the Gallai-Edmonds decomposition: finding the canonical partition of an elementary graph, computing the maximum number of vertex disjoint paths between two subsets of vertices of a graph, and computing a minimal separator between two subsets of vertices of a graph. A close inspection of these variants (Sections 3.3 and 4 in \cite{CheriyanSICOMP97}) shows that they can also be implemented efficiently in the congested clique model, giving again algorithms with round complexity $O(n^{1-1/\omega})$.

\section*{Acknowkedgments}
The author is grateful to Arne Storjohann for precious help concerning the computation of the determinant and to anonymous reviewers for their comments. This work is supported by the Grant-in-Aid for Young Scientists~(A) No.~16H05853, the Grant-in-Aid for Scientific Research~(A) No.~16H01705, and the Grant-in-Aid for Scientific Research on Innovative Areas~No.~24106009 of the Japan Society for the Promotion of Science and the Ministry of Education, Culture, Sports, Science and Technology in Japan. 

%\bibliographystyle{plain}
%\bibliography{CCM}

\end{document}